\documentclass[conference]{IEEEtran}

\pdfoutput=1

\usepackage{amsmath}
\usepackage{amsthm}
\usepackage{amssymb}
\usepackage{epsfig}
\usepackage{subfigure}
\usepackage{multirow}
\newtheorem{thm}{Theorem}
\newtheorem{cor}{Corollary}
\newtheorem{lem}{Lemma}

\begin{document}
\title{Statistical End-to-end Performance Bounds for Networks under Long Memory FBM Cross Traffic}
\author{
\authorblockN{Amr Rizk and Markus Fidler}
\authorblockA{Institute of Communications Technology, Leibniz Universit\"{a}t Hannover}}
\maketitle
\begin{abstract}
Fractional Brownian motion (fBm) emerged as a useful model for self-similar and long-range dependent Internet traffic. Approximate performance measures are known from large deviations theory for single queuing systems with fBm through traffic. In this paper we derive end-to-end performance bounds for a through flow in a network of tandem queues under fBm cross traffic. To this end, we prove a rigorous sample path envelope for fBm that complements previous approximate results. We find that both approaches agree in their outcome that overflow probabilities for fBm traffic have a Weibullian tail. We employ the sample path envelope and the concept of leftover service curves to model the remaining service after scheduling fBm cross traffic at a system. Using composition results for tandem systems from the stochastic network calculus we derive end-to-end statistical performance bounds for individual flows in networks under fBm cross traffic. We discover that these bounds grow in $\mathcal{O}\bigl(n (\log n)^{\frac{1}{2-2H}} \bigr)$ for $n$ systems in series where $H$ is the Hurst parameter of the fBm cross traffic. We show numerical results on the impact of the variability and the correlation of fBm traffic on network performance.
\end{abstract}
\section{Introduction}
\label{sec:Introduction}
Since the beginning of the Internet classical queuing theory has been used as the decisive methodology to analyze many relevant problems in computer networking \cite{kleinrock75}, most prominently to prove the efficiency of packet switching over circuit switching. The significance of queuing theory stems from explicit closed-form performance measures, e.g. backlog and delay, for queuing systems under Poisson traffic\footnote{Poisson traffic refers to exponential packet inter-arrival and service times.}, such as the M$\mid$M$\mid$1 single server system with infinite buffer. In practice results obtained for systems with infinite buffer are frequently transferred with reasonable precision to systems with finite buffer, e.g. for buffer sizing \cite{appenzeller:sizingbuffers}. The outstanding features of queuing theory that are essential for analysis of networks are (\textit{i}) multiplexing and demultiplexing of flows\footnote{Multiplexing independent Poisson flows leads to a Poisson flow as well as blind demultiplexing of a Poisson flow leads to Poisson flows.}, i.e. routing, and (\textit{ii}) product form queuing networks where constituent queues can be analyzed as if in isolation.

While it has long been known that Internet traffic sources do not fulfill the memoryless property assumed by the Poisson traffic model, it has been argued that the aggregate of a large number of multiplexed flows will tend towards Poisson. This assumption is backed up by the fact that the sum of many independent Bernoulli trials converges to a Poisson random variable. In the mid 90's, the groundbreaking discovery from extensive measurements was, however, that aggregate Internet traffic possesses significant long-range dependence (LRD) as well as statistical self-similarity \cite{leland94,paxon95,crovella97,willinger97,feldmann:iptraffic}. A possible mathematical justification for these traffic properties arises from the superposition of many independent on-off sources with heavy-tailed on and off periods \cite{taqqu97}. This model matches file size distributions on storage systems \cite{crovella97,willinger97}.

A random process that captures self-similarity is fractional Brownian motion (fBm) \cite{mandelbrot68}. FBm is characterized by the Hurst parameter $H$ where fBm with $H\in \left(\frac{1}{2},1\right)$ is LRD and widely accepted in the literature \cite{leland94,norros94,norros95} as a useful model for aggregate Internet traffic. The cases $H=\frac{1}{2}$ and $H\in \left(0,\frac{1}{2}\right)$ correspond to standard Brownian motion and short-range dependent fBm, respectively. Using theories including effective bandwidths and large deviations significant results have been derived for the performance of single queuing systems fed with LRD fBm traffic \cite{norros95,duffield94,massoulie99,kelly96}. Regarding the analysis of networks these theories can, however, not carry the outstanding properties of queuing theory forward to fBm.

Properties similar to (\textit{i}) and (\textit{ii}) from queuing theory have, however, been established in the deterministic \cite{chang00,leboudec:networkcalculus} and stochastic \cite{chang00,li07,burchard06,ciucu06,fidler06,jiang:stochasticnetworkcalculus} network calculus. The network calculus uses the concept of service curves \cite{parekh:processorsharing1,sariowan:servicecurves,cruz:qosmanagement} to characterize the service provided by queuing systems. Leftover service curves are used to analyze the effects of scheduling, i.e. multiplexing and demultiplexing cross traffic (\textit{i}). Service curves of tandem systems are composed by convolution enabling the analysis of networks  (\textit{ii}). Finally, arrival envelopes are used as a model of traffic flows \cite{cruz:networkdelaycalculus,sidi93,starobinski:stochasticallyboundedburstiness,boorstyn:effectiveenvelopes,Yin02,li07} to derive performance bounds on backlog and delay for service curve systems.

In this paper we derive end-to-end statistical performance bounds for tandem systems under LRD fBm cross traffic. We contribute a rigorous sample path envelope for fBm traffic that complements an approximate envelope that follows from a known asymptotic backlog bound. Both envelopes agree in the Weibullian decay of overflow probabilities. These envelopes are the basis for network analysis using the stochastic network calculus. Owing to the concept of leftover service curves we quantify the effects of fBm cross traffic on the performance of through flows and find that the correlation of the cross traffic has severe impact. Finally, we derive end-to-end performance bounds for a through flow that traverses $n$ tandem systems each under fBm cross traffic. We show that these bounds grow in $\mathcal{O}\bigl(n (\log n)^{\frac{1}{2-2H}} \bigr)$. Our finding compares to \cite{ciucu06} where end-to-end performance bounds in $\mathcal{O}(n \log n)$ have been derived for traffic with exponentially as opposed to Weibull bounded burstiness. For $H = \frac{1}{2}$ we recover the previous result.

The remainder of this paper is structured as follows. In Sect. \ref{sec:RelatedWork} we review related work on the application of fBm for performance analysis of computer networks. In Sect. \ref{sec:SamplePathEnvelopeForLongRangeDependentFBmTraffic} we provide a sample path envelope for LRD fBm traffic based on its moment generating function. In Sect. \ref{sec:LeftoverServiceCurveUnderFBMCrossTraffic} we derive a leftover service curve under fBm cross traffic. In Sect. \ref{sec:EndToEndConcatenation} we present end-to-end service curves and statistical performance bounds for tandem systems. Sect. \ref{sec:Conclusion} gives brief conclusions.
\section{Related Work on FBM Queuing Systems}
\label{sec:RelatedWork}
FBm is a widely accepted model for self-similar Internet traffic \cite{leland94,norros94,norros95}. It can be short- or long-range dependent conditional on its Hurst parameter $H \in (0,\frac{1}{2})$ or $H \in (\frac{1}{2},1)$, respectively. An fBm process $Z(t)$ has stationary Gaussian increments and the following basic properties: $Z(0)=0$, $\mathsf{E}[Z(t)] = 0$, and $\mathsf{E}[Z(t)^2] = \sigma^2 t^{2H}$ for all $t \ge 0$ where $\sigma > 0$ is its standard deviation at $t=1$. The increment process of fBm, also called fractional Gaussian noise (fGn), has positive as well as negative increments. While this property of fGn is not matched by real network traffic it avoids excessive complexity and keeps the model solvable \cite{kelly96}. The long-range dependence of fBm for $H \in (\frac{1}{2},1)$ is established by the infinite sum $\sum_t v(t) = \infty$ of the auto-covariance of the increments $v(t) \approx \sigma^2 H(2H-1) t^{2H-2}$ as $t \rightarrow \infty$.

In the sequel cumulative arrivals from a traffic source in an interval $[\tau,t)$ are denoted $A(\tau,t)$. Shorthand notation $A(t)$ is used for $A(0,t)$. The arrivals of an fBm traffic source are modeled as the superposition of a mean rate $\lambda$ and an fBm process $Z(t)$
\begin{equation}
A(t) = \lambda t + Z(t).
\label{eq:aggr_A}
\end{equation}

The challenge of the fBm traffic model in case of LRD is that its variance $\sigma^2 t^{2H}$ grows superlinearly in $t$, i.e. the traffic is heavily bursty with burst periods that are more likely to be sustained for a long time. These properties make the calculation of performance bounds for fBm traffic hard.

The backlog process $B(t)$ at a lossless work-conserving constant rate server with capacity $C$ is described by Reich's equation, see e.g. \cite{kumar04},
\begin{equation*}
B(t) = \sup_{\tau \in [0,t]} \left\{A(\tau,t) - C (t-\tau)\right\}.
\end{equation*}
The difficulty behind the analysis of a statistical bound $b$ for the steady state backlog $B$, i.e. letting $t \rightarrow \infty$, is to find the value $\tau^*$ that achieves the supremum in
\begin{equation}
\mathsf{P}[B > b] = \mathsf{P} \biggl[ \sup_{\tau \in [0,t]} \{ A(\tau,t) - C(t-\tau) \} > b \biggr]
\label{eq:backlogbound}
\end{equation}
since $\tau^*$ is a random variable, see \cite{li07} for explanation.

Large deviations theory is frequently used to analyze the asymptotic decay rate of the overflow probability of a backlog bound \cite{glynn94,duffield94,massoulie99,chang00}. The asymptotic calculation makes use of the principle of the largest term stating
\begin{equation}
\mathsf{P}[B > b] \approx \sup_{\tau \in [0,t]} \mathsf{P} \left[A(\tau,t) - C(t-\tau) > b\right]
\label{eqn:change_P_sup}
\end{equation}
where the term on the right hand side strictly provides only a lower bound. For fBm traffic (\ref{eq:aggr_A}) at a server with capacity $C$ the following asymptotic holds for the decay rate of the overflow probability \cite{duffield94}
\begin{equation*}
\lim_{b\rightarrow \infty} \frac{\log \mathsf{P} [ B > b ]}{b^{2-2H}} = -\inf_{k>0} k^{2H-2}\frac{\left(k+C-\lambda\right)^2}{2}.
\end{equation*}
It follows that $\lim_{b \rightarrow \infty} \mathsf{P} [B>b] = \varepsilon_a$ where \cite{duffield94}
\begin{equation}
\varepsilon_a = \exp \Biggl(-\frac{1}{2\sigma^2}\biggl(\frac{C-\lambda}{H}\biggr)^{2H} \biggl( \frac{b}{1-H} \biggr)^{2-2H} \Biggr).
\label{eq:norros_bnd}
\end{equation}
The resulting overflow probability has a Weibull tail that simplifies to an exponential distribution for the special case $H=\frac{1}{2}$. The backlog bound was proven to be logarithmically asymptotical generally and exact for $H=\frac{1}{2}$.

The large deviations result (\ref{eq:norros_bnd}) agrees with a solution deduced for the largest term (\ref{eqn:change_P_sup}) in \cite{norros94,norros95}. The derivation makes use of the Gaussian distribution of the increments of fBm and yields the approximation
\begin{equation*}
\mathsf{P}\left[B>b\right] \approx \sup_{t\geq0}\bar{\Phi}\biggl(\frac{(C-\lambda)t+b}{\sigma t^H}\biggr)
\end{equation*}
where $\bar{\Phi}(x)=P(Z(1)>x)$ is the complementary cumulative distribution function of a Gaussian random variable, i.e. the increment of fBm. After maximizing over $t$ the backlog bound approximation in \cite{norros94} is $\mathsf{P} [B>b] \approx \varepsilon_a$ where $\varepsilon_a$ is identical to (\ref{eq:norros_bnd}). A comprehensible introduction covering the derivation of this bound can also be found in \cite{grimm08}.

The proof of the large deviations theory builds on the G\"artner-Ellis condition which establishes a direct relation to the effective bandwidth of a traffic flow \cite{duffield94}. The theory of effective bandwidths, e.g. \cite{kelly96,chang00}, is a major tool for the analysis of traffic flows as it gives a measure for resource requirements at different time scales. The effective bandwidth of a flow $\alpha(\theta,t) = \frac{1}{\theta t} \log \mathsf{E}\left[e^{\theta A(t)}\right]$ lies between its average and peak rate depending on the parameter $\theta > 0$. For fBm it holds that
\begin{equation}
\alpha(\theta,t) = \lambda + \frac{\theta \sigma^2}{2} t^{2H-1} .
\label{eq:fbmeffbw}
\end{equation}
In case of $H \in (\frac{1}{2},1)$ the effective bandwidth of fBm traffic exhibits a continuous growth in $t$ due to LRD \cite{kelly96}.

In \cite{li07} a connection between effective bandwidths and effective envelopes is established. In contrast to asymptotic results for large buffers from large deviations theory, effective envelopes  in conjunction with the stochastic network calculus \cite{boorstyn:effectiveenvelopes,Yin02,li07,burchard06,ciucu06,jiang:stochasticnetworkcalculus} can provide non-asymptotic performance bounds. Moreover, recent stochastic network calculus provides methods for derivation of stochastic leftover service curves as well as for composition of tandem systems. Effective envelopes $E(t-\tau)$ are statistical upper bounds of the cumulative arrivals $A(\tau,t)$ of the form
\begin{equation*}
\mathsf{P} [A(\tau,t) - E(t-\tau) > 0 ] \leq \varepsilon_p.
\end{equation*}
An envelope for fBm traffic is derived in \cite{mayor:fbmtimescale,fonseca00,li07} as
\begin{equation}
E(t) = \lambda t + \sqrt{-2\log \varepsilon_p} \sigma t^H .
\label{eq:pointwisefbmenvelope}
\end{equation}

The definition of effective envelope is point-wise in the sense that it can be violated at each point in time with overflow probability $\varepsilon_p$. Applying the approximation by the largest term (\ref{eqn:change_P_sup}) the envelope (\ref{eq:pointwisefbmenvelope}) is used in \cite{fonseca00} to recover the backlog bound from large deviations theory (\ref{eq:norros_bnd}).

In contrast, the derivation of performance bounds using the stochastic network calculus builds on sample path arguments, such as the backlog bound (\ref{eq:backlogbound}), and requires a bound for $A(\tau,t)$ for all $\tau \in [0,t]$, i.e. a sample path envelope of the form
\begin{equation}
\mathsf{P}\biggl[ \sup_{\tau \in [0,t]} \{ A(\tau,t) - E(t-\tau) \} > 0 \biggr] \leq \varepsilon_s.
\label{eq:samplepathenvelope}
\end{equation}
Such sample path envelopes are constructed in \cite{li07} using Boole's inequality under the assumption of a time scale limit $T$, i.e. by summing the constant point-wise overflow probabilities $\varepsilon_p$ over all $t \in [0,T]$. The time scale in this context can be regarded as a constraint on the duration of busy periods. In case of fBm the duration of busy periods has, however, been found to grow extremely fast with $H$ \cite{mayor:fbmtimescale}.

Methods for construction of sample path envelopes that do not require a priori assumptions on the relevant time scale have been developed in \cite{cruz:qosmanagement,Yin02,ciucu06}. The general approach is to use a point-wise envelope with parameter $b$
\begin{equation}
\mathsf{P}\left[A(\tau,t) - E(t-\tau) > b \right] \leq \varepsilon_p(b).
\label{eq:pointwiseenvelopeprofile}
\end{equation}
that has a decaying and integrable overflow profile $\varepsilon_p(b)$, i.e. $\int_{0}^{\infty} \varepsilon_p(b)db$ is finite. Typically, when constructing a sample path envelope $b$ is substituted by a slack rate $\varrho \cdot (t-\tau)$. The slack rate relaxes the envelope such that $\varepsilon_p$ decreases with increasing interval width $(t-\tau)$. Finally, taking Boole's inequality over all $\tau$ to derive the sample path overflow probability $\varepsilon_s$ (\ref{eq:samplepathenvelope}) translates to integrating $\int_{0}^{t} \varepsilon_p(\varrho (t-\tau)) d\tau$ that remains finite for all $t \ge 0$ including $t \rightarrow \infty$.

The construction of a sample path envelope for fBm traffic is, however, not straightforward since no simple overflow profile exists, see Sect. \ref{sec:SamplePathEnvelopeForLongRangeDependentFBmTraffic}. A simplifying approach is proposed in \cite{Yin02} where it is argued that the Weibull tail (\ref{eq:norros_bnd}) implies an envelope for fBm traffic. Such an envelope is, however, based on the approximation by the largest term (\ref{eqn:change_P_sup}). A rigorous sample path envelope for fBm as well as end-to-end performance bounds under fBm cross traffic have not been derived.
\section{Gamma Bound for FBM Sample Paths}
\label{sec:SamplePathEnvelopeForLongRangeDependentFBmTraffic}
In this section we derive a sample path envelope for fBm traffic. The envelope is an essential prerequisite for application of the stochastic network calculus. It is used to derive leftover service curves as well as end-to-end performance bounds in the following sections. In the remainder of this paper we consider a discrete time model, i.e. time $t$ is a dimensionless counter of time slots each of fixed duration. Correspondingly, for the fBm traffic model (\ref{eq:aggr_A}) the rate $\lambda$ is given in bits per time slot and the increment $Z(t)$ in bits. We use subscripted $\varepsilon_p$, $\varepsilon_s$, and $\varepsilon_a$ to denote point-wise, sample path, and approximate or asymptotic overflow probabilities, respectively, see Sect. \ref{sec:RelatedWork}.
\subsection{Sample Path Envelope}
The difficulty of deriving a sample path envelope for fBm traffic is due to the intended integrability of the point-wise overflow probability $\int_0^{\infty}\varepsilon_p(t)dt$ as introduced in Sect. \ref{sec:RelatedWork}. To this end, we consider the point-wise envelope (\ref{eq:pointwisefbmenvelope}) with a time-dependent overflow probability $\varepsilon_p(t)$
\begin{equation}
E(t) = \lambda t + \sqrt{-2\log \varepsilon_p(t)} \sigma t^H.
\label{eq:samplepathfbmenvelope}
\end{equation}
This envelope can be obtained from Chernoff's bound. It constitutes the natural shape of an fBm envelope in the sense that it is the optimal solution that can be derived in this way.

\begin{figure}
\centering
\subfigure[The parameter $\beta$ relaxes the envelope. The case $\beta = 0$ coincides with (\ref{eq:pointwisefbmenvelope}).]{
\includegraphics[width=1.0\columnwidth]{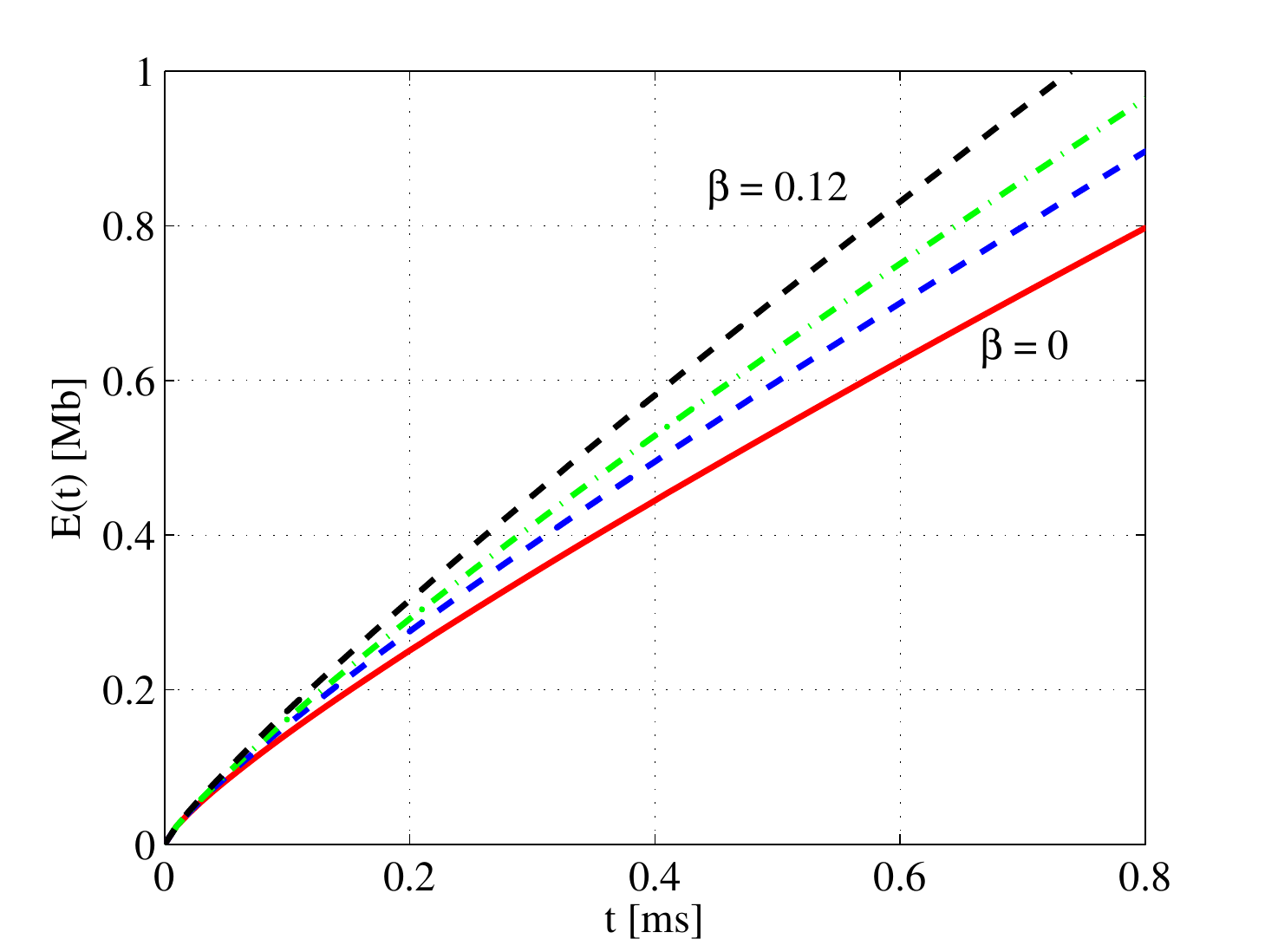}
\label{fig:env_dif_beta}
}
\subfigure[The point-wise overflow probability $\varepsilon_p(t) = \eta^{t^{2\beta}}$ decays faster with increasing $\beta$ but generally slower than exponentially since $\beta \in (0,1-H)$.]{
\includegraphics[width=1.0\columnwidth]{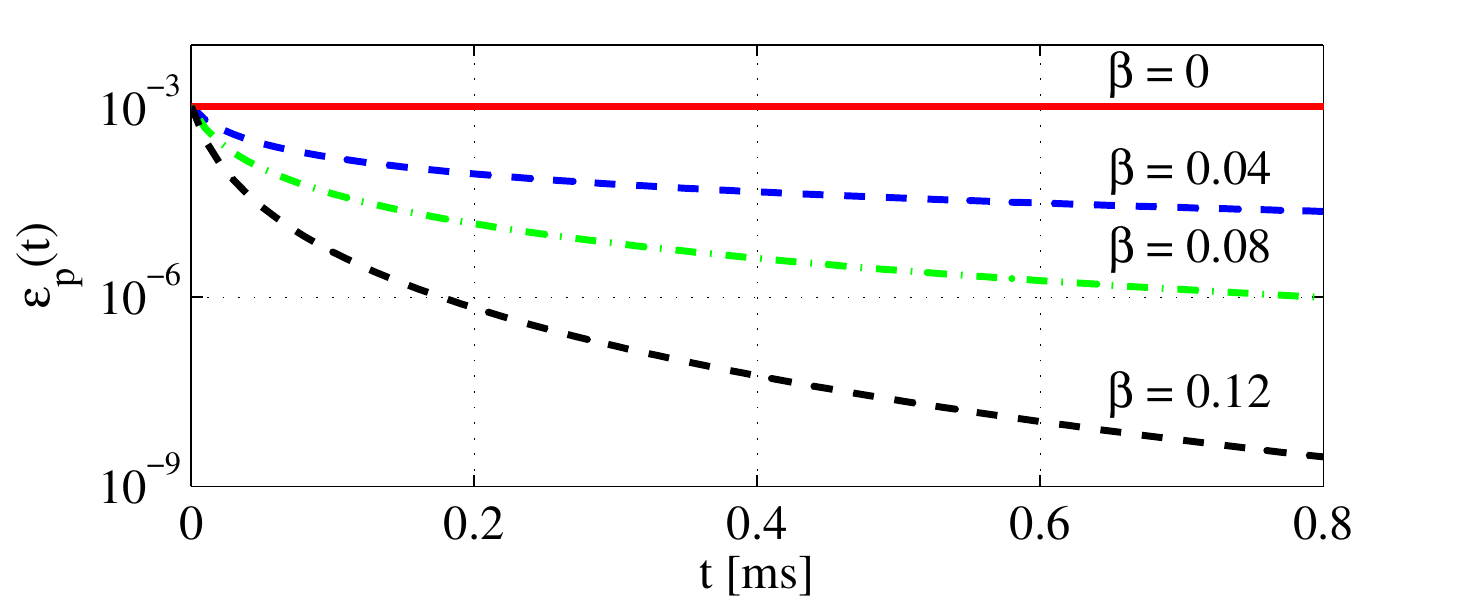}
\label{fig:pw_epsilon_diff_beta}
}
\caption{FBm envelopes according to (\ref{eq:samplepathfbmenvelope}).}
\label{fig:envelopes}
\end{figure}
Here, we have to choose the point-wise overflow probability $\varepsilon_p(t)$ in such a way that it is integrable. In fact twofold integrability is required to apply network calculus composition results for tandem systems in Sect. \ref{sec:EndToEndConcatenation}. At the same time $E(t)$ has to be at most linear in $t$ for $t \rightarrow \infty$ to be able to derive performance bounds, e.g. $\limsup_{t \rightarrow \infty} E(t)/t < C$ for a constant rate server with capacity $C$. This constrains $\varepsilon_p(t)$ to decay no faster than proportionally to $e^{-t^{2-2H}}$ for $t \rightarrow \infty$. We make the practicable choice\footnote{We note that other choices can be made, e.g. an envelope with linear slack rate $\varrho$, i.e. $E(t) = (\lambda + \varrho) t + \sqrt{-2\log \eta} \sigma t^H$, follows at the cost of significant complexity if we set $\varepsilon_p(t) = \exp{\bigl(-(\varrho t^{1-H} + \sqrt{-2 \log \eta} \sigma)^2 / (2 \sigma^2) \bigr)}.$} $\varepsilon_p(t) = \eta^{t^{2\beta}}$ with slack parameter $\beta \in (0,1-H)$ and $\eta \in (0,1)$ to derive an envelope for fBm sample paths. The integrability of $\varepsilon_p(t)$ is established by Lem. \ref{lem:gamma} in the appendix. For the special case $\beta = 0$ the parameter $\eta$ is the overflow probability of the point-wise envelope (\ref{eq:pointwisefbmenvelope}), i.e. $\varepsilon_p = \eta$ if $\beta = 0$.

Fig. \ref{fig:envelopes} evaluates different fBm envelopes. The fBm traffic parameters are $\lambda=0.5$, $\sigma=0.5$, and $H = 0.7$, i.e. the traffic is LRD as frequently observed for the Internet \cite{leland94,paxon95,crovella97,willinger97,feldmann:iptraffic}. To give the parameters a physical meaning we choose the time slot of the discrete time model to be 10 $\mu$s that is the transmission time of a 1250 Byte, respectively, 10 kb sized packet on a 1 Gb/s link. Accordingly, on a time slot basis $\lambda$ is 5 kb that translates to 0.5 Gb/s.

Fig. \ref{fig:env_dif_beta} compares fBm envelopes (\ref{eq:samplepathfbmenvelope}) with $\varepsilon_p = \eta^{t^{2\beta}}$ that are slackened by $\beta=$ 0.04, 0.08, and 0.12, respectively, to the envelope (\ref{eq:pointwisefbmenvelope}), i.e. the special case where $\beta=0$. The corresponding point-wise overflow probability $\varepsilon_p(t)$ is shown in Fig. \ref{fig:pw_epsilon_diff_beta}. Clearly, the overflow probability decays faster in case of larger $\beta$ whereas it remains constant and hence non-integrable if $\beta=0$. Generally, the overflow probability decays slower than exponentially for any choice of $\beta \in (0,1-H)$.
\begin{thm}[\textbf{FBM Sample Path Envelope}]
\label{theorem1}
Given fBm traffic with mean rate $\lambda$, standard deviation $\sigma$, and LRD Hurst parameter $H \in (\frac{1}{2},1)$.
\begin{equation*}
E(t) = \lambda t + \sqrt{-2 \log \eta} \sigma t^{H+\beta}
\end{equation*}
satisfies the definition of sample path envelope (\ref{eq:samplepathenvelope}) with overflow probability
\begin{equation*}
\varepsilon_s = \frac{\Gamma(\frac{1}{2\beta})}{2\beta (-\log \eta)^{\frac{1}{2\beta}}}
\end{equation*}
where $\beta \in (0,1-H)$ and $\eta \in (0,1)$ are free parameters.
\end{thm}
\begin{proof}
Assuming stationarity of the arrivals and letting $t \rightarrow \infty$ the sample path envelope (\ref{eq:samplepathenvelope}) can be written as
\begin{equation}
\mathsf{P} \biggl[ \sup_{\tau \ge 0} \{A(\tau) - E(\tau)\} > 0 \biggr] \le \varepsilon_s.
\label{eq:samplepath}
\end{equation}
The overflow probability is defined by a union of events
\begin{equation*}
\mathsf{P} \biggl[ \sup_{\tau \ge 0} \{A(\tau) - E(\tau)\} >0 \biggr] = \mathsf{P} \biggl[\bigcup_{\tau = 0}^{\infty} \{ A(\tau) - E(\tau) > 0 \} \biggr]
\end{equation*}
that can be estimated by application of Boole's inequality
\begin{equation*}
\mathsf{P} \biggl[\bigcup_{\tau = 0}^{\infty} \{ A(\tau) > E(\tau) \} \biggr] \le \sum_{\tau = 1}^{\infty} \mathsf{P} [ A(\tau) > E(\tau)]
\end{equation*}
where we used the fact that the overflow probability at $\tau = 0$ is trivially zero since by definition $A(0) = E(0) = 0$.

Applying Chernoff's bound and using the effective bandwidth of fBm (\ref{eq:fbmeffbw}) we have for any $\theta > 0$ that
\begin{equation*}
\mathsf{P}[A(\tau) \!>\! E(\tau)] \le e^{-\theta E(\tau)} \mathsf{E}\bigl[e^{\theta A(\tau)}\bigr] = e^{-\theta E(\tau)} e^{\theta \lambda \tau + \frac{\theta^2 \sigma^2}{2} \tau^{2H}}\!\!.
\end{equation*}
Inserting $E(\tau) = \lambda \tau + \sqrt{-2 \log \eta} \sigma \tau^{H+\beta}$ and minimizing over $\theta > 0$ yields $\theta =\frac{1}{\sigma} \sqrt{-2 \log \eta} \tau^{\beta - H}$ and by insertion
\begin{equation*}
\mathsf{P}[A(\tau) > E(\tau)] \le \eta^{\tau^{2\beta}}.
\end{equation*}
Summing the point-wise overflow probabilities gives
\begin{equation*}
\sum_{\tau=1}^{\infty} \eta^{\tau^{2\beta}} \le \int_{0}^{\infty} \eta^{\tau^{2\beta}} d\tau = \frac{\Gamma(\frac{1}{2\beta})}{2\beta (-\log \eta)^{\frac{1}{2\beta}}}
\end{equation*}
where we used that $\eta^{\tau^{2\beta}}$ is monotonically decreasing in $\tau$ to estimate each summand indexed by $\tau$ by an integral over $(\tau-1,\tau]$. Finally, we applied Lem. \ref{lem:gamma} from the appendix.
\end{proof}

\begin{figure}
\centering
\subfigure[Point-wise overflow probability $\varepsilon_p(t) = \eta^{t^{2\beta}}$ vs. simulation results. The overflow probabilities are conservative due to the use of Chernoff's bound.]{
\includegraphics[width=1.0\columnwidth]{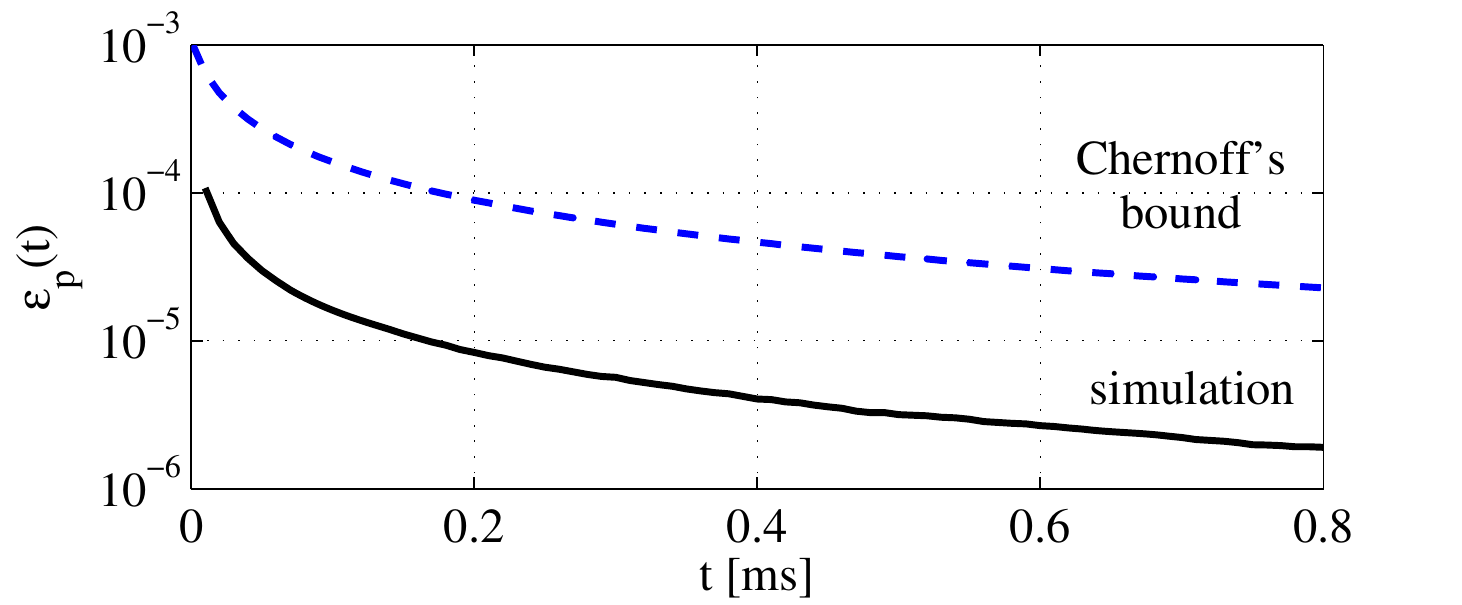}
\label{fig:pw_epsilon_diff_beta_sim}
}
\subfigure[Sample path overflow probability $\varepsilon_s$ vs. simulation results for sample paths of length $t$. The overflow probabilities are strict upper bounds due to the sample path argument using Boole's inequality. For $t \rightarrow \infty$ the sample path overflow probability stays finite and converges to the result obtained from Th. \ref{theorem1}.]{
\includegraphics[width=1.0\columnwidth]{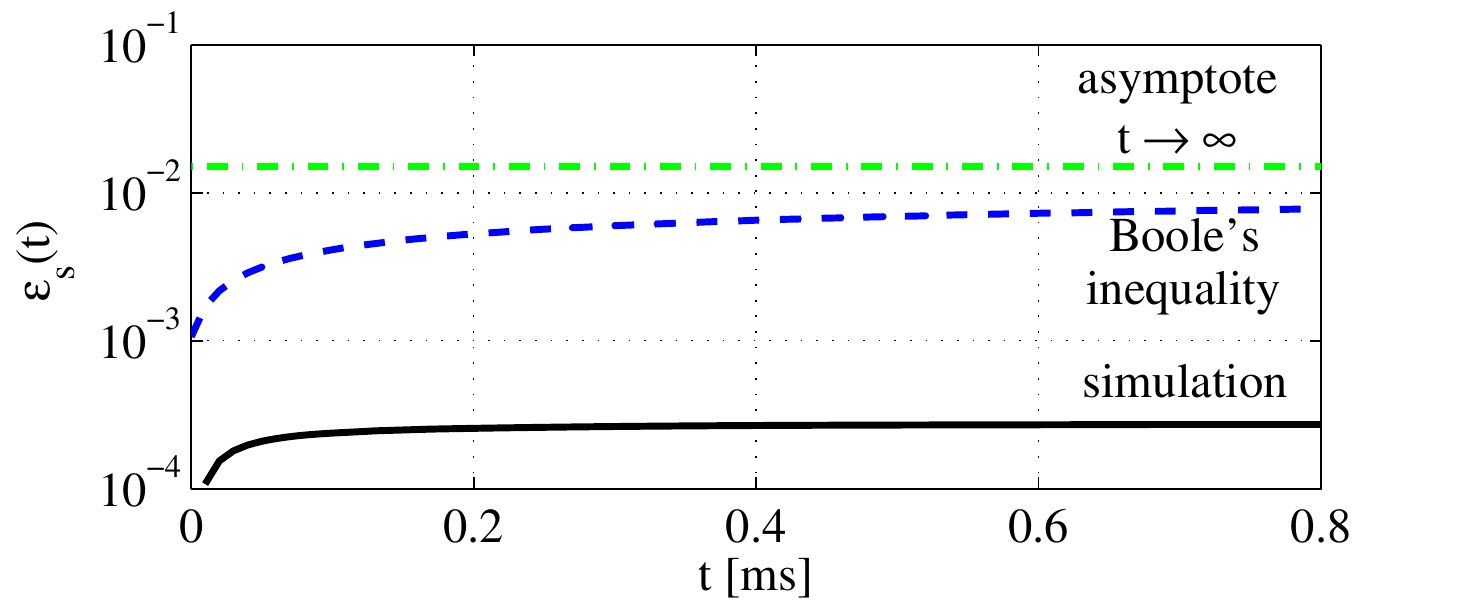}
\label{fig:sp_epsilon_diff_beta_sim}
}
\caption{Overflow probabilities for envelopes from Th. \ref{theorem1} where $\beta = 0.04$.}
\end{figure}
Fig. \ref{fig:pw_epsilon_diff_beta_sim} compares the point-wise overflow probability $\varepsilon_p(t) = \eta^{t^{2\beta}}$ of the envelope from Fig. \ref{fig:env_dif_beta} for $\beta = 0.04$ with simulation results obtained from $10^9$ fBm sample paths generated with Matlab. The overflow probabilities are computed from Chernoff's bound and hence are conservative. Fig. \ref{fig:sp_epsilon_diff_beta_sim} shows the overflow probability for sample paths of length $t$. The simulation results are obtained by counting the number of sample paths that violate the envelope at least once in $[0,t]$. To derive the analytical bound we sum the point-wise overflow probabilities over all $\tau \in [0,t]$. The dashed horizontal line shows the overflow probability of the sample path envelope from Th. \ref{theorem1} that holds for sample paths of any length $t \ge 0$, i.e. the sum of the point-wise overflow probabilities stays finite for all $t$ and converges to the dashed line for $t \rightarrow \infty$.
\subsection{Performance Bounds}
\label{sec:performancebounds}
We use the sample path envelope from Th. \ref{theorem1} to derive performance bounds for a server fed with fBm traffic. Owing to a known duality between backlog bounds and envelopes we find affine envelopes for fBm traffic along the way, see Sect. \ref{sec:affineenvelopes}.
\begin{thm}[\textbf{Backlog and Delay Bound}]
\label{thm:backlog}
Consider a lossless work-conserving constant rate server with capacity $C$ fed with fBm traffic as in Th. \ref{theorem1} where $C > \lambda$. The steady state backlog $B$ is bounded by $b$ subject to the overflow probability $\varepsilon_s$
\begin{equation*}
\mathsf{P}[B > b] \le \varepsilon_s = \frac{\Gamma(\frac{1}{2\beta})}{2\beta (-\log\eta)^{\frac{1}{2\beta}}}
\end{equation*}
where $\beta \in (0,1-H)$ is a free parameter and
\begin{equation*}
\eta = \exp \Biggl(\! -\frac{1}{2\sigma^2} \! \left(\frac{C-\lambda}{H+\beta}\right)^{\!2(H+\beta)} \!\! \left(\frac{b}{1-(H+\beta)}\right)^{\!2-2(H+\beta)} \! \Biggr) .
\end{equation*}
The steady state delay under first-come first-serve (fcfs) scheduling $W$ is bounded by $\mathsf{P} [W > b/C] \le \varepsilon_s$.
\end{thm}
\begin{proof}
Assuming stationarity and letting $t \rightarrow \infty$ we obtain the steady state backlog from (\ref{eq:backlogbound}) as
\begin{equation*}
\mathsf{P}[B > b] = \mathsf{P} \biggl[ \sup_{\tau \ge 0} \{ A(\tau) - C \tau \} > b \biggr].
\end{equation*}
The expression is a special case of sample path envelope (\ref{eq:samplepath}). Define $E(t)$ as in Th. \ref{theorem1}. If $E(t) \le b + Ct$ for all $t \ge 0$ then
\begin{equation*}
\mathsf{P}[B > b] \le \mathsf{P} \biggl[ \sup_{\tau \ge 0} \{A(\tau) - E(\tau)\} > 0 \biggr] \le \frac{\Gamma(\frac{1}{2\beta})}{2\beta (-\log\eta)^{\frac{1}{2\beta}}} .
\end{equation*}

Given $b$ and $C$ we derive the largest envelope that satisfies the constraint $E(t) \le b + Ct$ for all $t \ge 0$. To this end, we first find $t = \tau^*$ that minimizes the vertical distance between $E(t)$ and $b+Ct$. Hence, $\tau^*$ is the solution of $\partial E(t)/\partial t = C$. Then, we choose the parameter $\eta \in (0,1)$ such that $E(t)$ and $b+Ct$ are tangent to each other at $\tau^*$, i.e. $E(\tau^*) = b+C\tau^*$.
From the first requirement we derive $\tau^*$ as
\begin{equation}
\tau^* = \left(\frac{\sqrt{-2\log\eta} \sigma (H+\beta)}{C-\lambda}\right)^{\frac{1}{1-(H+\beta)}}.
\label{eq:overflowtimescale}
\end{equation}
By insertion of $\tau^*$ into the condition $E(\tau) = b+C\tau$ the optimal parameter $\eta$ follows as given in Th. \ref{thm:backlog}.

The delay bound can be derived as the maximum horizontal distance of $E(t)$ and $Ct$ using the same basic steps.
\end{proof}

An approach related to Th. \ref{thm:backlog} has been employed in \cite{fonseca00} using, however, the point-wise fBm envelope (\ref{eq:pointwisefbmenvelope}) and the principle of the largest term (\ref{eqn:change_P_sup}) to compute an approximate backlog bound. In contrast, our sample path envelope does not resort to this approximation and establishes rigorous upper bounds.

We recover previous results for the special case where $\beta = 0$. In this case the optimal parameter $\eta$ derived in Th. \ref{thm:backlog} is equal to the overflow probability of the known asymptotic backlog bound (\ref{eq:norros_bnd}), i.e. $\eta = \varepsilon_a$. Recall that (\ref{eq:norros_bnd}) has been derived using either the Gaussian distribution of the fBm increments \cite{norros94,norros95}, large deviations theory \cite{duffield94}, or the point-wise envelope (\ref{eq:pointwisefbmenvelope}) \cite{fonseca00}. All these derivations are based on the approximation by the largest term (\ref{eqn:change_P_sup}). Similarly, $\eta$ is the point-wise overflow probability of the fBm envelope $E(t)$ from (\ref{eq:pointwisefbmenvelope}), i.e. $\eta = \varepsilon_p$, where $E(t)$ is the largest envelope that is smaller than $b+Ct$ for all $t \ge 0$. Furthermore, $E(t)$ is tangent at $\tau^*$ such that $b + Ct$ generally has an overflow probability smaller than $\eta$ except at $\tau^*$ where it attains its maximum that equals $\eta$. Hence, the supremum in (\ref{eqn:change_P_sup}) is attained at $\tau^*$ and the boundary point at $\tau^*$ is the most probable point for violation of $b + Ct$, i.e. for buffer overflow. From (\ref{eq:overflowtimescale}) the most probable time scale for violation evaluates by insertion of $\eta$ from Th. \ref{thm:backlog} and $\beta = 0$ to
\begin{equation*}
\tau^* = \biggl(\frac{b}{C-\lambda}\biggr) \biggl(\frac{H}{1-H}\biggr)
\end{equation*}
recovering the result from \cite{norros94}. Related probabilistic bounds for the duration of busy periods $\tau'$ at a constant rate server fed with fBm traffic have been derived in \cite{mayor:fbmtimescale} as the point in time where the point-wise envelope (\ref{eq:pointwisefbmenvelope}) equals $C t$
\begin{equation*}
\tau' = \biggl( \frac{\sqrt{-2 \log \varepsilon_p} \sigma}{C-\lambda} \biggr)^{\frac{1}{1-H}}.
\end{equation*}
It is interesting to observe the large impact of $H$ on the relevant time scales that is unfavorable for the use time scale bounds.

\begin{figure}
\centering
\includegraphics[width=1.0\columnwidth]{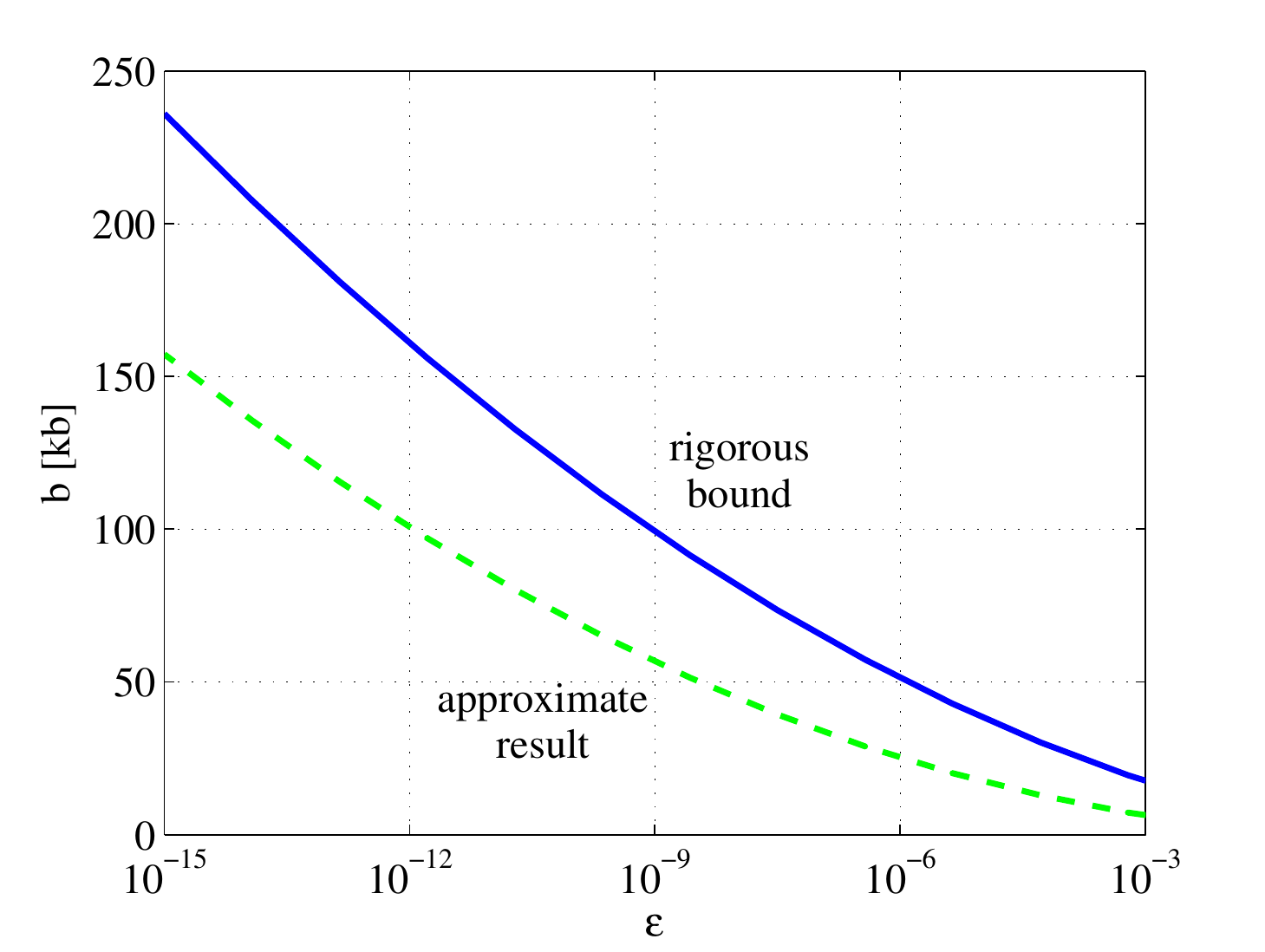}
\caption{Rigorous backlog bound from Th. \ref{thm:backlog} in comparison to the approximation resp. asymptotic result from (\ref{eq:norros_bnd}) for fBm traffic with parameters $\lambda = 0.5$ Gb/s, $\sigma = 0.25$ Gb/s, and $H=0.75$ at a server with capacity $C = 1$ Gb/s. Both bounds decay slower than exponentially due to LRD.}
\label{fig:backlbnd_diff_eps}
\end{figure}
Fig. \ref{fig:backlbnd_diff_eps} displays backlog bounds for a server with capacity $C = 1$ Gb/s fed with fBm traffic with parameters $\lambda = 0.5$ Gb/s, $\sigma = 0.25$ Gb/s, and $H=0.75$. The parameter set agrees with observations in \cite{norros95} and will be used in the sequel unless mentioned otherwise. The overflow probabilities are obtained from Th. \ref{thm:backlog} using sample path arguments (solid line) and from (\ref{eq:norros_bnd}) using the principle of the largest term (dashed line), respectively. We optimize parameter $\beta$ numerically. Both results agree with each other regarding the slower than exponential decay that is due to LRD. The absolute values differ, however, by a factor of about 1.75 at $\varepsilon = 10^{-9}$ due to the fact that Th. \ref{thm:backlog} provides a rigorous upper bound whereas (\ref{eq:norros_bnd}) is an approximation. The relative difference of the two results decreases for smaller $\varepsilon$.

\begin{figure}
\centering
\subfigure[Impact of mean rate $\lambda$]{
\includegraphics[width=0.46\columnwidth]{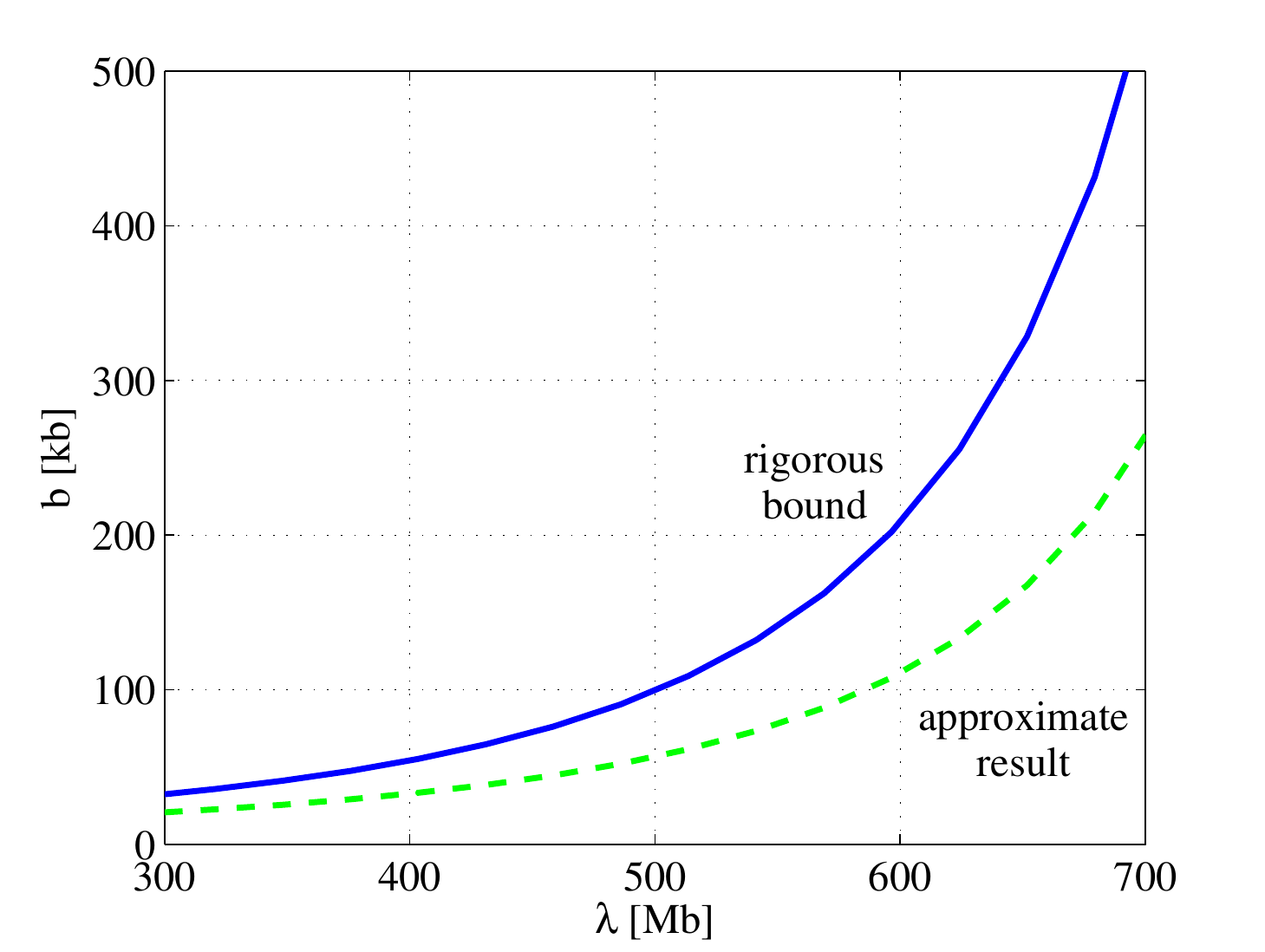}
\label{fig:backlbnd_diff_lambda}
}
\subfigure[Impact of standard deviation $\sigma$]{
\includegraphics[width=0.46\columnwidth]{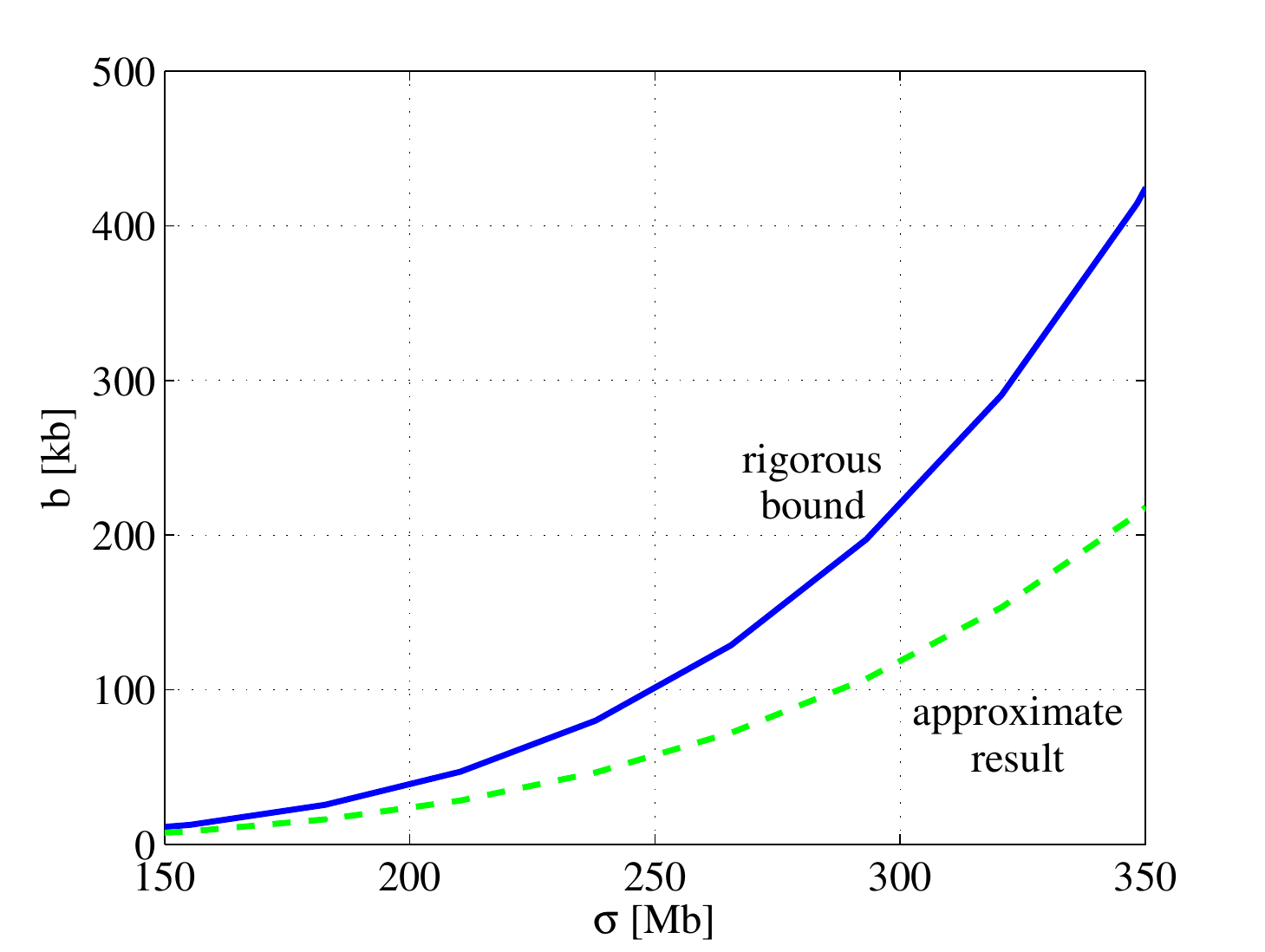}
\label{fig:backlbnd_diff_sigma}
}
\subfigure[Impact of $H$, moderate $H$]{
\includegraphics[width=0.46\columnwidth]{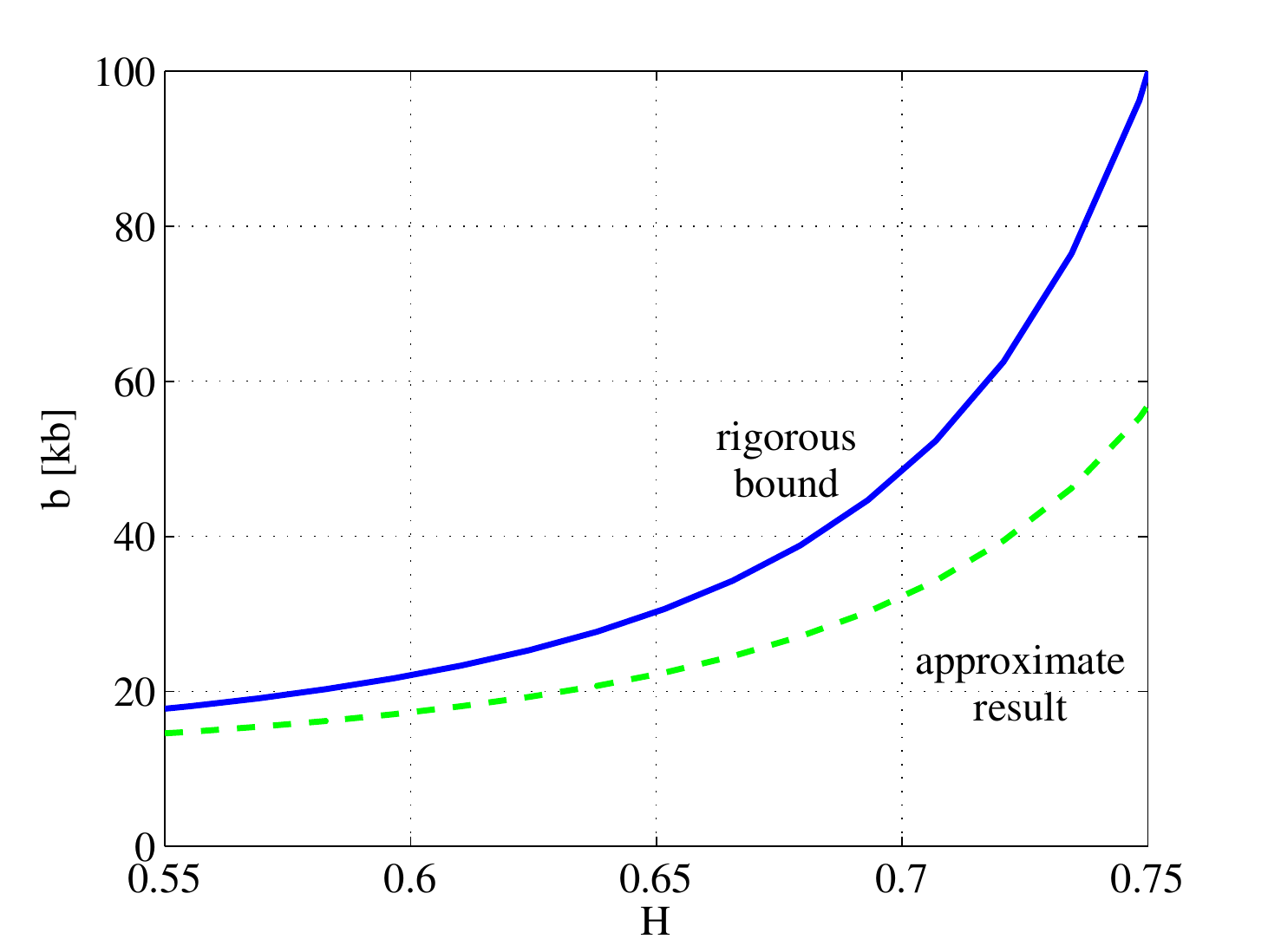}
\label{fig:backlbnd_diff_H_lin}
}
\subfigure[Impact of $H$, large $H$]{
\includegraphics[width=0.46\columnwidth]{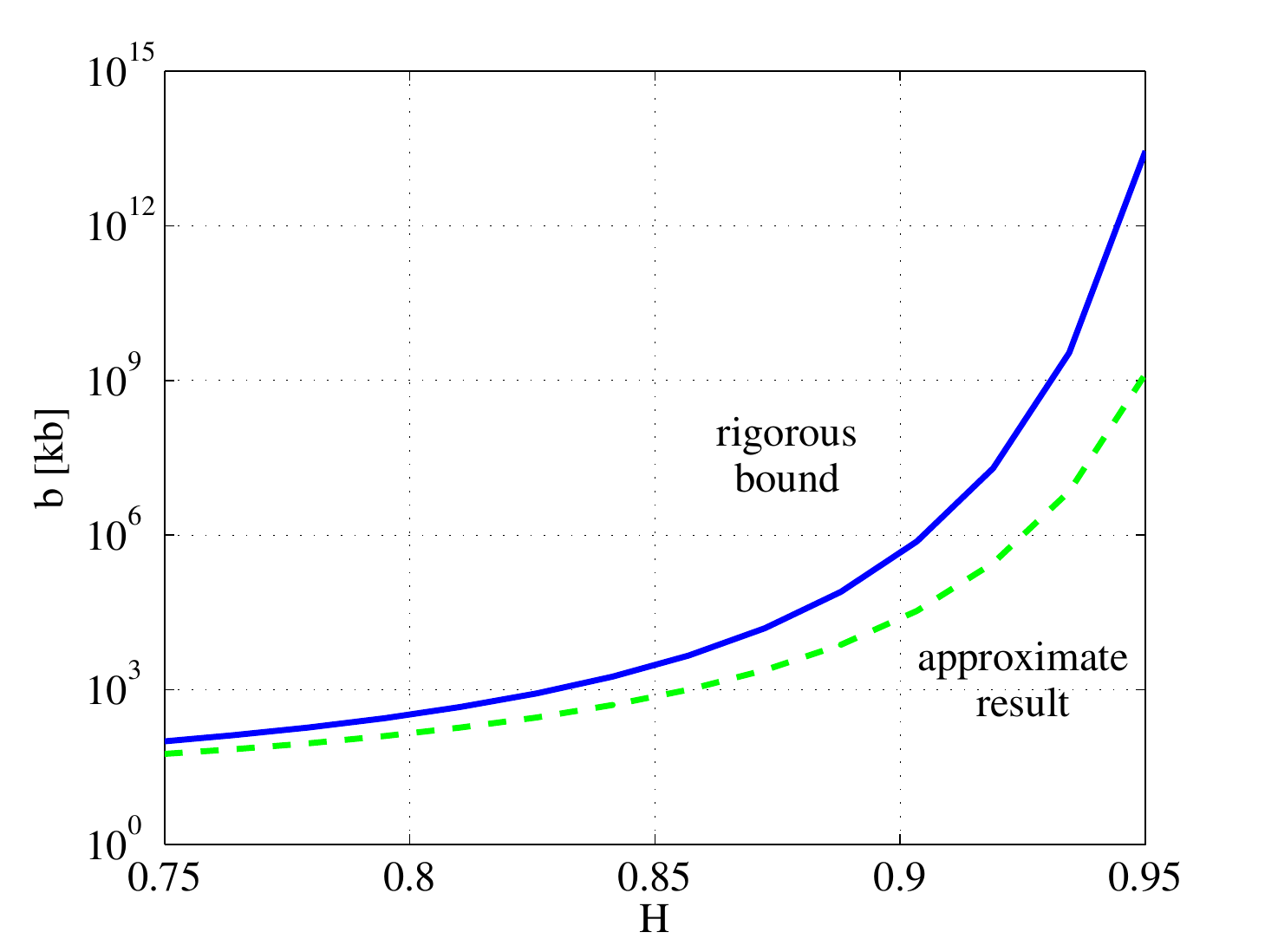}
\label{fig:backlbnd_diff_H_log}
}
\caption{Effect of different traffic parameters on backlog bounds as in Fig. \ref{fig:backlbnd_diff_eps} with overflow probability $\varepsilon = 10^{-9}$. Note the tremendous impact of the Hurst parameter, hence Fig. \ref{fig:backlbnd_diff_H_log} is on log scale.}
\label{fig:backlbnd_diff_pars}
\end{figure}
Fig. \ref{fig:backlbnd_diff_pars} shows how a variation of the different parameters of the fBm traffic affects the backlog bounds. Clearly, for larger $\lambda$, $\sigma$, or $H$ both backlog bounds increase. For moderate $H$ the rigorous and the approximate bounds agree well, whereas large $H$ have a huge impact, hence Fig. \ref{fig:backlbnd_diff_H_log} is on log scale.

Fig. \ref{fig:backlbnd_diff_sigma} can also be interpreted as showing the effect of statistical multiplexing of fBm flows. Given $m$ statistically independent fBm flows each normalized by $m$, i.e. with rate $\lambda/m$ and standard deviation $\sigma/m$ and identical Hurst parameter, the aggregate traffic of all $m$ flows has rate $\lambda$ but standard deviation $\sigma/\sqrt{m}$. Statistical multiplexing of fBm flows is discussed in detail in \cite{fonseca00}.

\subsection{Weibullian Decay of Overflow Probabilities}
\label{sec:nearoptimalbeta}
In this subsection we show that overflow probabilities of fBm sample path bounds have a Weibull tail in $b$, i.e. $\log \varepsilon_s \sim -b^{2-2H}$. To this end, we minimize $\varepsilon_s$ from Th. \ref{thm:backlog} over $\beta$.

To obtain the results shown in Fig. \ref{fig:backlbnd_diff_eps} and Fig. \ref{fig:backlbnd_diff_pars} from Th. \ref{thm:backlog} we optimized the free parameter $\beta \in (0, 1-H)$ numerically. We note that the optimal choice of $\beta$ has significant impact, e.g. on the decay of $\varepsilon$ with $b$. Typically, we find that the optimal $\beta$ is small compared to one. In this case a useful approximation of the Gamma bound in Th. \ref{theorem1} and Th. \ref{thm:backlog} is
\begin{equation}
\tilde{\varepsilon}_s = \frac{\sqrt{\pi}}{\sqrt{\beta}(2e \beta (-\log \eta))^{\frac{1}{2\beta}}} \;\, \text{ for } \;\, \beta \ll 1
\label{eq:stirlingapprox}
\end{equation}
using Stirling's formula $\Gamma(x) \approx \sqrt{2 \pi / x} \; (x/e)^x$ for $x \gg 1$. The approximation (\ref{eq:stirlingapprox}) is exact in the limit $\beta \rightarrow 0$.

In the remainder of this subsection we optimize the free parameter $\beta \in (0,1-H)$. We use (\ref{eq:stirlingapprox}) and minor simplifications for small $\beta$ to derive a near optimal solution $\beta^*$. We emphasize that Th. \ref{thm:backlog} holds for any $\beta \in (0,1-H)$ where we approximate $\beta^*$ that minimizes the bound, i.e. inserting $\beta^*$ into Th. \ref{thm:backlog} yields a rigorous upper bound that is close to its minimal solution.

Assuming $\beta \ll (1-H)$ and approximating $1-H-\beta$ by $1-H$ we compute the derivative of (\ref{eq:stirlingapprox}) and solve $\partial \tilde{\varepsilon}_s/\partial \beta = 0$ for $\beta$. we find that the minimum of (\ref{eq:stirlingapprox}) is approached at $\beta = - W(1/(2 \log \varepsilon_a))$ where $\varepsilon_a$ coincides with (\ref{eq:norros_bnd}). $W(z)$ denotes Lambert's W function that is the inverse of $z = x e^x$. It is real-valued for $\varepsilon_a < e^{-e/2}$. Since $\beta$ is assumed to be small a good approximation of the optimal solution is
\begin{equation*}
\beta^* = \frac{1}{2 (-\log\varepsilon_a)}
\end{equation*}
where we estimate the Lambert W function by a linear segment. Note that $\beta^*$ decreases with $b$. We define the quotient
\begin{equation}
\chi = \biggl(\frac{H^{H} (1-H)^{1-H}}{(H+\beta)^{H+\beta} (1-(H+\beta))^{1-(H+\beta)}}\biggr)^2
\label{eq:chi}
\end{equation}
that is in $[\frac{1}{4},1]$ for $H \in (\frac{1}{2},1)$, $\beta \in (0,1-H)$ and approaches 1 for small $\beta$. Inserting $\beta^*$ into (\ref{eq:stirlingapprox}) we find the closed form
\begin{equation}
\tilde{\varepsilon}_s = \frac{b}{C-\lambda} \, \varepsilon_a^{1+\log\chi} \sqrt{2 \pi (- \log \varepsilon_a)}
\label{eq:samplepathapprox}
\end{equation}
as a near optimal solution of Th. \ref{thm:backlog}. We note that (\ref{eq:samplepathapprox}) matches the numerically optimized results in Fig. \ref{fig:backlbnd_diff_eps} and \ref{fig:backlbnd_diff_pars} almost perfectly which is why we omit reproducing similar graphs.

Phrasing $\varepsilon_s$ as a function of $\varepsilon_a$ enables us to support significant conclusions on the dimensioning of networks from the largest term approximation by sample path arguments. An important example is the relation of spare capacity $C - \lambda$ and buffer size $b$. While for $H=\frac{1}{2}$ halving $C-\lambda$ requires doubling $b$ to achieve constant $\varepsilon_a$ (\ref{eq:norros_bnd}) the tradeoff deteriorates for large $H$ \cite{norros95}. Under LRD spare capacity becomes more important and buffering much less efficient supporting current arguments for reducing router buffers \cite{appenzeller:sizingbuffers}.

Inserting $\varepsilon_a$ from (\ref{eq:norros_bnd}) into (\ref{eq:samplepathapprox}), bounding $\chi$ from below by a constant that is close to 1 for $\beta \ll 1-H$, and using positive constants $c_i$ yields after simplification that
\begin{equation*}
\log \tilde{\varepsilon}_s = - c_1 b^{2-2H} + \log(c_2 b^{2-H}) + c_3.
\end{equation*}
Based on sample path arguments the expression reflects the decay rate of the largest term approximation (\ref{eq:norros_bnd}).

To illustrate the tail behavior we compare the overflow probability of backlog bounds for fBm traffic to traffic with exponentially bounded burstiness (EBB) \cite{sidi93}. For generation of EBB traffic we use a discrete time Markov model with two states: on and off. In the off state (state 1) no traffic is generated and in the on state (state 2) traffic is generated with peak rate $P$. The steady state probability of the on state is $p_{\text{on}} = p_{12}/(p_{12} + p_{21})$ where $p_{ij}$ are the state transition probabilities from state $i$ to state $j$. The mean rate becomes $\lambda = p_{\text{on}} P$. As \cite{ciucu06} we characterize the burstiness of an on-off source by the average time for two state changes to occur $T = 1/p_{12} + 1/p_{21}$. The aggregate of $m$ on-off sources has a point-wise envelope (\ref{eq:pointwiseenvelopeprofile}) given as $E(t) = m \rho(\theta) t$ with overflow profile $\varepsilon_p(b) = e^{-\theta b}$, see \cite{ciucu06}. The envelope rate $\rho(\theta)$ of a single discrete time on-off source is given in \cite{chang00} as
\begin{equation*}
\frac{1}{\theta} \! \log \! \Biggl( \! \frac{p_{11} \! + \! p_{22} e^{\theta P} \! + \! \sqrt{(p_{11} \! + \! p_{22} e^{\theta P})^2 \! - \! 4(p_{11} \! + \! p_{22} \! - \! 1)e^{\theta P}}}{2} \Biggr)
\end{equation*}
for $\theta > 0$ where $\rho(\theta)$ ranges between the mean and the peak rate of the source. Employing the envelope a backlog bound follows from (\ref{eq:backlogbound}) by application of Boole's inequality and integration of the overflow profile, see Sect. \ref{sec:RelatedWork} and \cite{ciucu06} for details. Given aggregate traffic generated by $m$ on-off sources fed into a server with capacity $C$ the overflow probability $\mathsf{P} [B > b] \le \varepsilon_s$ of the backlog bound $b$ becomes
\begin{equation}
\varepsilon_s = \int_0^{\infty} e^{-\theta (b + (C - m\rho(\theta)) \tau)} d\tau = \frac{e^{-\theta b}}{\theta (C-m\rho(\theta))}
\label{eq:ebbbacklog}
\end{equation}
for any $\theta > 0$ that satisfies $m \rho(\theta) < C$.

For Fig. \ref{fig:fbmvsebbtail} we employ the traffic parameters given in Tab. \ref{tab:parametersthrough}. We use a traffic aggregate consisting of 100 EBB flows. In case of fBm we use only a single flow. Recall that fBm is typically used as a model for aggregate traffic, e.g. it has been related to the superposition of on-off sources with heavy-tailed on and off periods \cite{taqqu97,crovella97,willinger97}. In contrast the on-off sources used for EBB traffic have geometrically distributed on and off periods. The fBm traffic has the same parameters as used for Fig. \ref{fig:backlbnd_diff_eps} and Fig. \ref{fig:backlbnd_diff_pars}.

Fig. \ref{fig:fbmvsebbtail} clearly shows the Weibullian decay of the fBm backlog bound (\ref{eq:samplepathapprox}) with $\log \varepsilon_s \sim -b^{2-2H}$ that becomes slower with increasing $H$, i.e. with more pronounced LRD even very large buffers are filled with a certain probability. In contrast, the decay is exponentially fast in case of EBB traffic where the correlation parameter $T$ affects only the slope. From Fig. \ref{fig:fbmvsebbtail} we conclude that reasonable performance guarantees for EBB as well as fBm traffic exist if a certain violation probability is tolerated. A near-deterministic service with very small violation probability is, however, only practicable for EBB traffic. Due to the slow tail decay it is very costly to reduce the overflow probability for fBm traffic with LRD. We mention that the approximation (\ref{eq:samplepathapprox}) as before agrees well with Th. \ref{thm:backlog}. For EBB we optimized the parameter $\theta$ numerically.
\begin{table}
\begin{center}
\caption{Traffic parameters used for Fig. \ref{fig:fbmvsebbtail}.}
\label{tab:parametersthrough}
\begin{tabular}{|c|cccc|} \hline
type & \# sources & mean rate & variability & correlation \\ \hline \hline
EBB & $m = 100$ & $\lambda = 5$ Mb/s & $P = 5 \lambda$ & $T$  \\
fBm & $m=1$ & $\lambda = 500$ Mb/s & $\sigma = \lambda/2$ & $H$ \\ \hline
\end{tabular}
\end{center}
\end{table}
\begin{figure}
\begin{center}
\includegraphics[width=1.0\columnwidth]{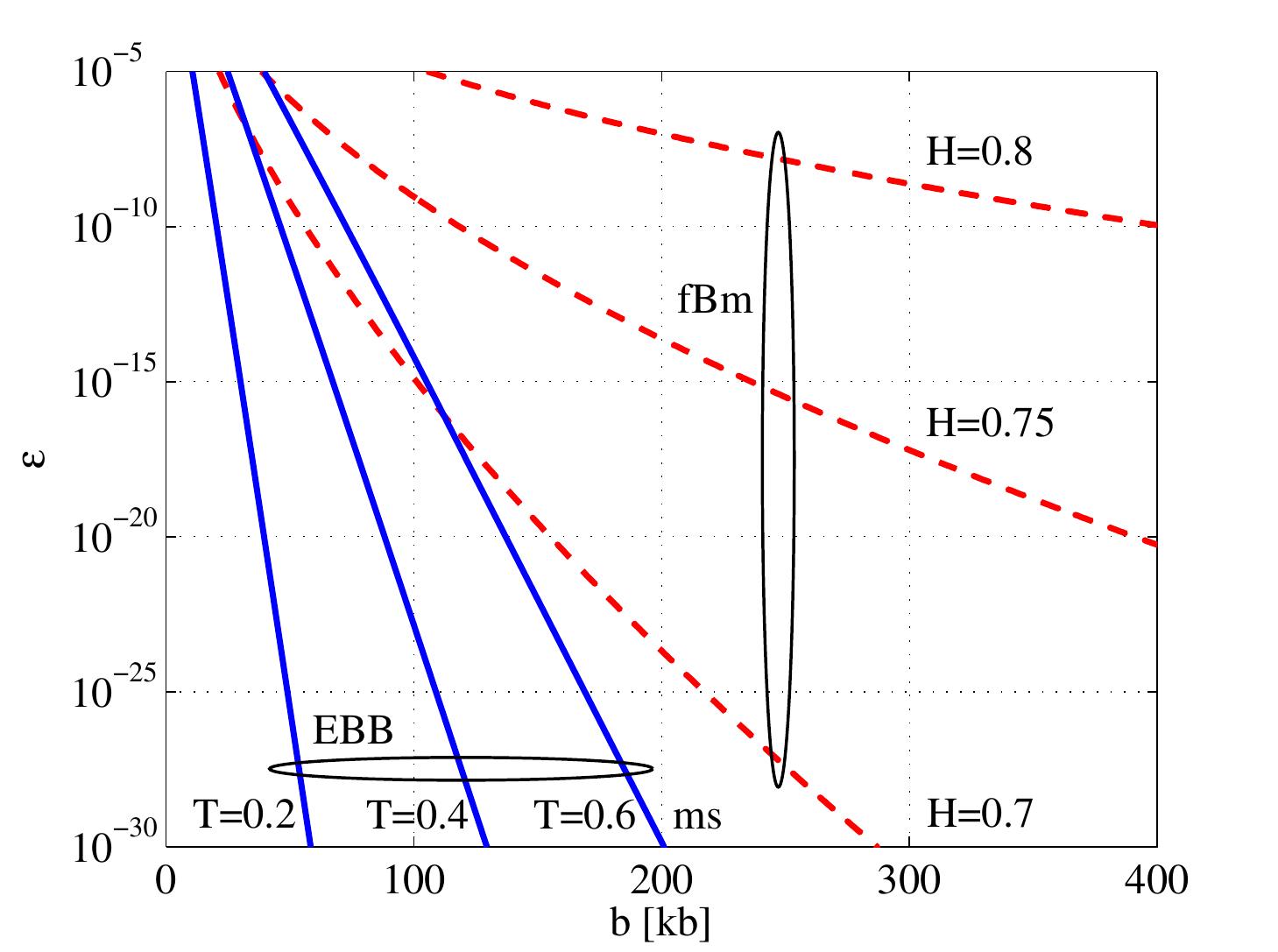}
\end{center}
\caption{Weibullian decay of the overflow probability for fBm traffic where $\log \varepsilon \sim -b^{2-2H}$ compared to the exponential decay for EBB traffic.}
\label{fig:fbmvsebbtail}
\end{figure}
\subsection{Affine Envelopes}
\label{sec:affineenvelopes}
We conclude this section with a corollary on affine fBm envelopes that follow from the backlog bounds. The proof of Th. \ref{thm:backlog} already shows a close link between arrival envelopes and backlog bounds. This relation has been elaborated in \cite{Yin02} where a traffic model referred to as generalized Stochastically Bounded Burstiness (gSBB) is defined. The gSBB traffic characterization uses an affine sample path envelope with defined overflow profile. A general definition of sample-path envelope with overflow profile $\varepsilon(b)$ \cite{cruz:qosmanagement,ciucu06} extends (\ref{eq:samplepathenvelope}) to
\begin{equation}
\mathsf{P}\biggl[ \sup_{\tau \in [0,t]} \{ A(\tau,t) - E(t-\tau) \} >  b \biggr] \leq \varepsilon_s(b).
\label{eq:samplepathenvelope_profile}
\end{equation}
The gSBB envelope is the special case where $E(t) = r t$. The finding in \cite{Yin02} is that given a server with capacity $C$ and statistical backlog bound $b$ with overflow probability $\mathsf{P}[B > b] \le \varepsilon(b)$ the arrivals of the system are gSBB and satisfy the definition of sample path envelope (\ref{eq:samplepathenvelope_profile}) with $E(t) = C t$ and identical overflow profile $\varepsilon(b)$. It is argued in \cite{Yin02} that a gSBB envelope for fBm traffic follows immediately from (\ref{eq:norros_bnd}).
\begin{cor}[\textbf{Affine FBM Envelopes}]
\label{cor:affineenvelopes}
Given fBm traffic $E(t) = r t$ is a sample path envelope (\ref{eq:samplepathenvelope_profile}) with overflow profile
\begin{equation*}
\varepsilon_s(b) = \frac{\Gamma(\frac{1}{2\beta})}{2\beta \vartheta^{\frac{1}{2\beta}}} \, b^{-\frac{1-(H+\beta)}{\beta}}
\end{equation*}
where $\beta \in (0,1-H)$ is a free parameter and
\begin{equation*}
\vartheta = \frac{1}{2\sigma^2} \! \left(\frac{r-\lambda}{H+\beta}\right)^{\!2(H+\beta)} \!\! \left(\frac{1}{1-(H+\beta)}\right)^{\!2-2(H+\beta)} .
\end{equation*}
An approximate overflow profile is $\varepsilon_a(b) = e^{-\upsilon b^{2-2H}}$ where
\begin{equation*}
\upsilon = \frac{1}{2\sigma^2}\biggl(\frac{r-\lambda}{H}\biggr)^{2H} \biggl( \frac{1}{1-H} \biggr)^{2-2H}.
\end{equation*}
\end{cor}
The first statement of the corollary uses our backlog bound from Th. \ref{thm:backlog} that is based on the sample path envelope from Th. \ref{theorem1}. The second statement uses (\ref{eq:norros_bnd}) as proposed in \cite{Yin02}. Note, however, that (\ref{eq:norros_bnd}) is based on the approximation by the largest term (\ref{eqn:change_P_sup}), i.e. it does not provide a sample path envelope.

The derived sample path envelopes are an essential building block for application of the stochastic network calculus. In Sect. \ref{sec:LeftoverServiceCurveUnderFBMCrossTraffic} we will use the envelopes to derive a leftover service curve to analyze systems where fBm cross traffic is multiplexed, scheduled, and de-multiplexed afterwards. In Sect. \ref{sec:EndToEndConcatenation} we will compose these service curves by convolution to explore tandem systems, each under fBm cross traffic. We will use the affine envelopes established by Cor. \ref{cor:affineenvelopes}. The envelope from Th. \ref{theorem1} can be used in the same way and may yield tighter bounds at the cost of additional complexity.
\section{Leftover Service under FBM Cross Traffic}
\label{sec:LeftoverServiceCurveUnderFBMCrossTraffic}
The network calculus uses the concept of service curves $S(t)$ to model the service provided by a system. The service curve relates a system's departures $D(t)$ to its arrivals $A(t)$ \cite{sariowan:servicecurves}. A fundamental stochastic service curve is defined in \cite{ciucu06} and likewise in \cite{cruz:qosmanagement}
\begin{equation}
\mathsf{P}\biggl[D(t) < \inf_{\tau \in [0,t]} \{ A(\tau) + [S(t-\tau)-b ]_+ \} \biggr] \leq \varepsilon(b) \;
\label{eq:stochastic_srv_crv}
\end{equation}
where $[x]_+ = \max \{0,x\}$. The stochastic service curve is subject to a deficit profile $\varepsilon(b)$ that is decreasing in $b$. The inner operation $\inf_{\tau \in [0,t]} \{f(\tau) + g(t-\tau) \} := f \otimes g(t)$ is referred to as the convolution under the min-plus algebra.

So-called leftover service curves can effectively characterize the service that remains for a through flow at a system after scheduling cross traffic \cite{boorstyn:effectiveenvelopes,ciucu06}. A basic leftover service curve that does not make any assumptions about the order of scheduling can be deduced by subtracting the sample path envelope of the cross traffic from the service provided by the system. Given a server with capacity $C$ and cross traffic with sample path envelope $E(t)$ and overflow profile $\varepsilon_s(b)$ as in (\ref{eq:samplepathenvelope_profile}). A leftover service curve that satisfies (\ref{eq:stochastic_srv_crv}) with deficit profile $\varepsilon_s(b)$ is $S(t) = C t - E(t)$. The following corollary uses the affine fBm sample path envelope from Cor. \ref{cor:affineenvelopes} to characterize the leftover service under fBm cross traffic.
\begin{table}
\begin{center}
\caption{Traffic parameters used for Fig. \ref{fig:delaybnd_diff_throughtraffic}, Fig. \ref{fig:delaybnd_diff_end2end}, and Fig. \ref{fig:delayinc_diff_end2end}.}
\label{tab:parameters}
\begin{tabular}{|c|c|cccc|} \hline
traffic & type & \# sources & mean rate & variability & correlation \\ \hline \hline
\multirow{3}{*}{$\!\!\!$through$\!\!\!$} & CBR & $m=100$ & $\lambda = 2.5$ Mb/s & -- & -- \\
& EBB & $m = 100$ & $\lambda = 2.5$ Mb/s & $P = 5 \lambda$ & $T^{th}$  \\
& fBm & $m=1$ & $\lambda = 250$ Mb/s & $\sigma = \lambda/2$ & $H^{th}$ \\ \hline
\multirow{2}{*}{cross} & EBB & $m=100$ & $\lambda = 2.5$ Mb/s & $P = 5 \lambda$ & $T^{cr}$ \\
& fBm & $m=1$ & $\lambda = 250$ Mb/s & $\sigma = \lambda/2$ & $H^{cr}$ \\ \hline
\end{tabular}
\vspace{-2pt}
\end{center}
\end{table}
\begin{cor}[\textbf{FBM Leftover Service Curve}]
\label{cor:fbmleftoverservice}
Consider a server with capacity $C$ under fBm cross traffic. A service curve (\ref{eq:stochastic_srv_crv}) for through traffic is $S(t) = (C-r) t$ with deficit profile $\varepsilon(b)$. The deficit profile equals the overflow profile $\varepsilon_{s}(b)$ of the fBm envelope in Cor. \ref{cor:affineenvelopes}. It is approximated by $\varepsilon_{a}(b)$, Cor. \ref{cor:affineenvelopes}.
\end{cor}

The definitions of service curve and sample path arrival envelope facilitate the derivation of performance bounds. Given arrivals with sample path envelope $E(t)$ (\ref{eq:samplepathenvelope_profile}) at a system with service curve $S(t)$ (\ref{eq:stochastic_srv_crv}), the steady state virtual backlog $B$ is stochastically bounded by \cite{cruz:qosmanagement,ciucu06}
\begin{equation*}
\mathsf{P} \biggl[B > \sup_{\tau \ge 0} \left\{E(\tau) - S(\tau)\right\} + b \biggr] \leq \varepsilon(b)
\end{equation*}
where $\varepsilon(b) = \varepsilon^{th} \otimes \varepsilon^{cr}(b)$. We use superscripted $\varepsilon^{th}$ and $\varepsilon^{cr}$ to distinguish overflow profiles of through and cross traffic. The latter matches the deficit profile of the service curve. The fcfs waiting time $W$ is bounded by $\mathsf{P} [W > w] \leq \varepsilon(b)$ where
\begin{equation}
w = \inf \{\tau \ge 0: S(t+\tau)\geq E(t) + b \;\; \forall t \geq 0 \} .
\label{eq:horizontaldeviation}
\end{equation}
Backlog and delay can be visualized as the vertical, respectively, horizontal deviation of the arrival envelope and the service curve subject to the overflow and deficit profiles.

We derive performance bounds for a through flow under fBm cross traffic using the leftover service curve from Cor. \ref{cor:fbmleftoverservice}. We consider three fundamentally different types of through traffic: (\textit{a}) constant bit rate (CBR), (\textit{b}) traffic with exponentially bounded burstiness (EBB) \cite{sidi93}, and (\textit{c}) fBm with LRD, i.e. slower than exponential burstiness decay.

All three types of through traffic are configured to have the same mean rate $\lambda$. We generally use affine sample path envelopes (\ref{eq:samplepathenvelope_profile}). Trivially, the envelope of the CBR traffic is $E(t) = r t$ with $r=\lambda$ and overflow profile $\varepsilon_s(b) = 0$ for all $b \ge 0$. For generation of EBB traffic we use the discrete time Markov model introduced in Sect. \ref{sec:nearoptimalbeta}. A sample path envelope (\ref{eq:samplepathenvelope_profile}) for the aggregate traffic generated by $m$ on-off sources follows from the backlog bound (\ref{eq:ebbbacklog}) as $E(t) = r t$ with overflow profile $\varepsilon_s(b) = e^{-\theta b}/(\theta (r-m\rho(\theta)))$ for any $\theta > 0$ that satisfies $m\rho(\theta) < r$. For the fBm through traffic we use the envelope from Cor. \ref{cor:affineenvelopes}.

The delay bound follows from (\ref{eq:horizontaldeviation}) for $r^{th} + r^{cr} \le C$ as
\begin{equation*}
\mathsf{P} \biggl[W > \frac{b^{th} + b^{cr}}{C - r^{cr}} \biggr] \le \varepsilon^{th}(b^{th}) + \varepsilon^{cr}(b^{cr}).
\end{equation*}
Here, $b^{th}, r^{th}$ are the parameters of the through traffic envelope and $b^{cr}, r^{cr}$ are the parameters of the cross traffic envelope that determine the leftover service curve. Note that different sets of parameters can yield different violation probabilities for the same delay bound. To obtain the best possible result we optimize the parameters of the envelopes including the parameter $\beta$ for each traffic type numerically.

\begin{figure}
\centering
\includegraphics[width=1.0\columnwidth]{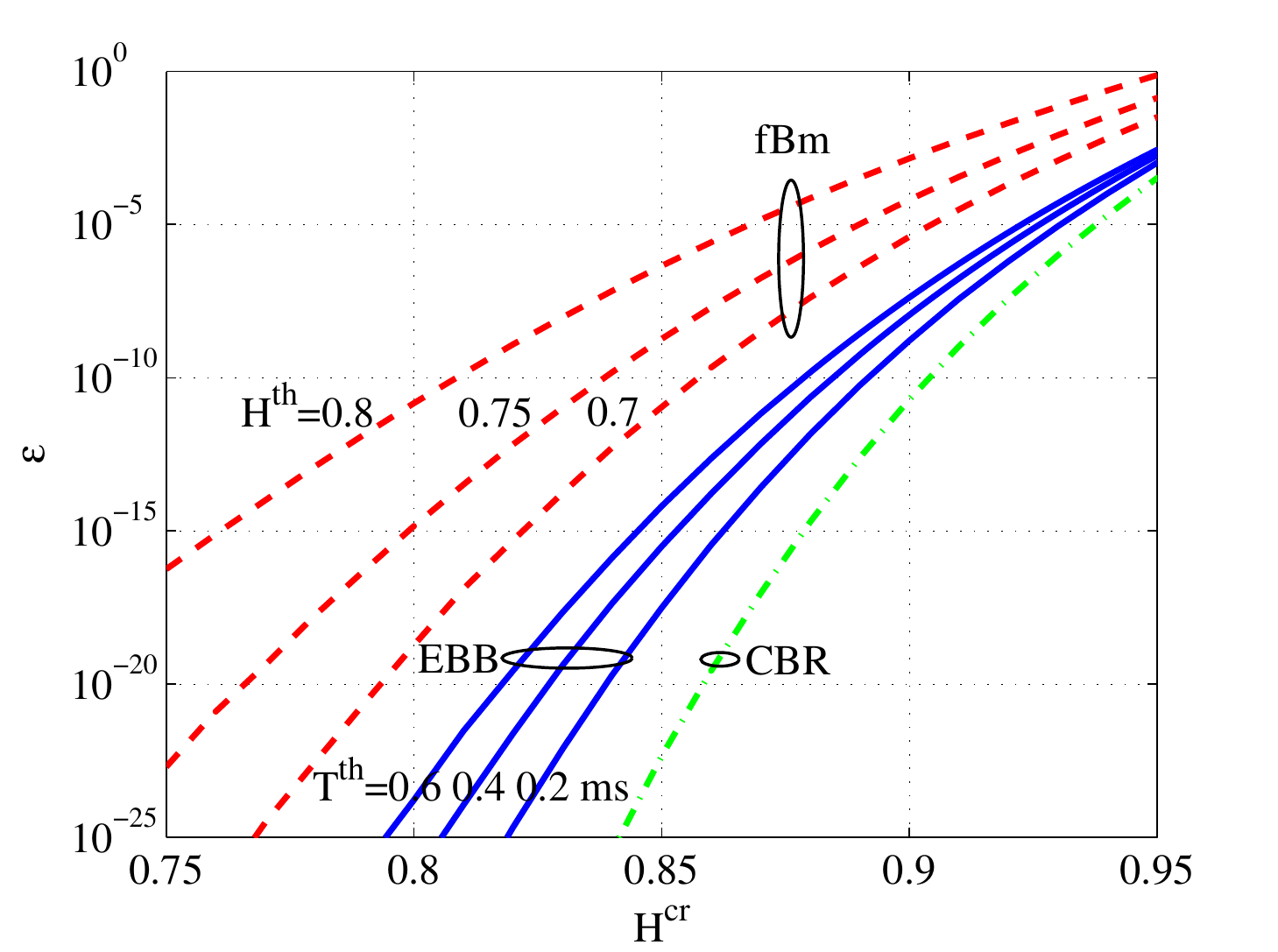}
\caption{Violation probability of a delay bound of 1 ms for different types of through traffic at a constant rate server with fBm cross traffic. The Hurst parameter of the cross traffic $H^{cr}$ has significant impact on the performance of the through traffic. While the burstiness of through traffic is crucial if $H^{cr}$ is small, the cross traffic dominates if $H^{cr}$ is large.}
\label{fig:delaybnd_diff_throughtraffic}
\end{figure}
In Fig. \ref{fig:delaybnd_diff_throughtraffic} we display the violation probability of a delay bound of 1 ms for CBR, EBB, and fBm through traffic at a server with capacity $C = 1$ Gb/s under fBm cross traffic. Through and cross traffic are configured to have the same mean rate. The traffic parameters are summarized in Tab. \ref{tab:parameters}. Fig. \ref{fig:delaybnd_diff_throughtraffic} shows the huge impact of the Hurst parameter $H^{cr}$ of the cross traffic on the performance of through flows. While type and burstiness of the through traffic largely determine the violation probability of the delay bound, the influence becomes much less pronounced for cross traffic with large $H^{cr}$, i.e. the LRD of the cross traffic becomes the dominating effect.
\section{End-to-End Performance Bounds}
\label{sec:EndToEndConcatenation}
The particular strength of the network calculus is its ability to characterize the end-to-end service of tandem systems. The service curves of individual systems can be composed by min-plus convolution into a network service curve. This network service curve models a whole network as if it were a single system. Hence, it facilitates the derivation of end-to-end performance bounds using single system results, e.g. (\ref{eq:horizontaldeviation}).
\subsection{Sample Path FBM Leftover Service Curve}
The easy composition of tandem systems in the network calculus is due to the associativity of min-plus convolution. For the special case of deterministic systems (\ref{eq:stochastic_srv_crv}) has overflow profile $\varepsilon(b) = 0$ for all $b \ge 0$ and reduces to $D(t) \ge A \otimes S(t)$ almost surely. Given a network of two systems with deterministic service curves $S^1(t)$ and $S^2(t)$ in series the departures of the first system $D^1(t) = A^1 \otimes S^1(t)$ are the arrivals to the second system $A^2(t) = D^1(t)$ and by recursive insertion it follows that $D^2(t) \ge A^2 \otimes S^2 (t) = (A^1 \otimes S^1) \otimes S^2(t)$. Using associativity of min-plus convolution $D^2 \ge A^1 \otimes S^{net}(t)$ where $S^{net} = S^1 \otimes S^2(t)$ is the network service curve.

The composition of stochastic service curves is, however, much more involved. The difficulty is due to the fact that a recursive insertion of the stochastic service curve (\ref{eq:stochastic_srv_crv}) is not possible since (\ref{eq:stochastic_srv_crv}) uses sample paths of the arrivals $A(t)$ but makes only a point-wise statement for the departures $D(t)$. To derive stochastic network service curves an extended definition of stochastic service curve that makes sample path guarantees for the departures is required. The problem does not occur in the deterministic case, see \cite{ciucu06,ciucu:networkcalculusscaling} for details.

A fundamental stochastic network service curve is derived in \cite{ciucu06,ciucu:networkcalculusscaling}. We make a marginal adaptation of the service curve to discrete time that is used in this work. For $n$ systems in series each with service curve $S^i(t)$ and deficit profile $\varepsilon^i(b)$ according to (\ref{eq:stochastic_srv_crv}) a network service curve is
\begin{equation}
S^{net}(t) = S^{1} \otimes S^{2}_{-\delta} \otimes \cdots \otimes S^{n}_{-(n-1)\delta}(t)
\label{eq:s_e2e}
\end{equation}
where $S_{-\delta}(t) = S(t)-\delta t$ and $\delta > 0$ is a free parameter. The network service curve satisfies (\ref{eq:stochastic_srv_crv}) with deficit profile
\begin{equation}
\varepsilon^{net}(b) = \varepsilon^1_{\delta} \otimes \varepsilon^2_{\delta} \otimes \cdots \varepsilon^{n-1}_{\delta} \otimes \varepsilon^n (b).
\label{eq:epsilon_e2e}
\end{equation}
The deficit profiles $\varepsilon^i_{\delta}(b)$ stem from an extended definition of stochastic service curve that makes guarantees for entire sample paths of the departures $D(t)$ as opposed to (\ref{eq:stochastic_srv_crv}) that only makes a point-wise statement. To derive such sample path guarantees \cite{ciucu06,ciucu:networkcalculusscaling} contributes an essential sample path service curve that is relaxed by parameter $\delta > 0$. The definition of sample path service curve states that
\begin{equation*}
\mathsf{P}\Biggl[\sup_{t \in [0,u]} \! \biggl\{ \inf_{\tau \in [0,t]} \bigl\{ A(\tau) + [S(t-\tau)-\delta (u-t)-b]_+ \bigr\} - D(t) \! \biggr\} \! > \! 0 \Biggr]
\end{equation*}
is smaller equal $\varepsilon_{\delta}(b)$. The relaxation by $\delta$ permits deriving the sample path deficit profile from the point-wise deficit profile of the departures using Boole's inequality as \cite{ciucu06,ciucu:networkcalculusscaling}
\begin{equation}
\varepsilon_{\delta}(b) = \frac{1}{\delta} \int_{b}^{\infty} \varepsilon(x) dx .
\label{eq:departuresamplepathprofile}
\end{equation}

We use the concept of network service curve to analyze tandem systems where fBm cross traffic is multiplexed and de-multiplexed. To this end, we derive the sample path deficit profile (\ref{eq:departuresamplepathprofile}) for a leftover service curve under fBm cross traffic.
\begin{cor}[\textbf{Sample Path FBM Leftover Service Curve}]
\label{cor:fbme2e}
Consider a server with capacity $C$ under fBm cross traffic as given in Cor. \ref{cor:affineenvelopes}. The leftover service curve $S(t) = (C - r) t$ relaxed by rate $\delta > 0$ has the sample path deficit profile (\ref{eq:departuresamplepathprofile})
\begin{equation*}
\varepsilon_{\delta}(b) = \frac{\Gamma(\frac{1}{2\beta})}{2\delta \vartheta^{\frac{1}{2\beta}} (1\!-\!(H\!+\!2\beta))} \, b^{-\frac{1-(H+2\beta)}{\beta}}
\end{equation*}
where $\beta \in \bigl(0,\frac{1-H}{2}\bigr)$ is a free parameter and $\vartheta$ from Cor. \ref{cor:affineenvelopes}.
\end{cor}
Cor. \ref{cor:fbme2e} follows by insertion of the deficit profile from Cor. \ref{cor:fbmleftoverservice}, respectively, Cor. \ref{cor:affineenvelopes} into (\ref{eq:departuresamplepathprofile}) as
\begin{equation*}
\varepsilon_{\delta}(b) = \frac{\Gamma(\frac{1}{2\beta})}{2\beta \delta \vartheta^{\frac{1}{2\beta}}} \int_{b}^{\infty} x^{-\frac{1-(H+\beta)}{\beta}} dx
\end{equation*}
that has a finite solution if $\beta < \frac{1-H}{2}$. Alternatively, the approximate overflow profile $\varepsilon_a(b)$ from Cor. \ref{cor:affineenvelopes} can be used to derive a similar result where $b^{2-2H}$ becomes, however, the second argument of an incomplete Gamma function.
\subsection{Scaling of End-to-end Performance Bounds}
\begin{figure}
\centering
\includegraphics[width=1.0\columnwidth]{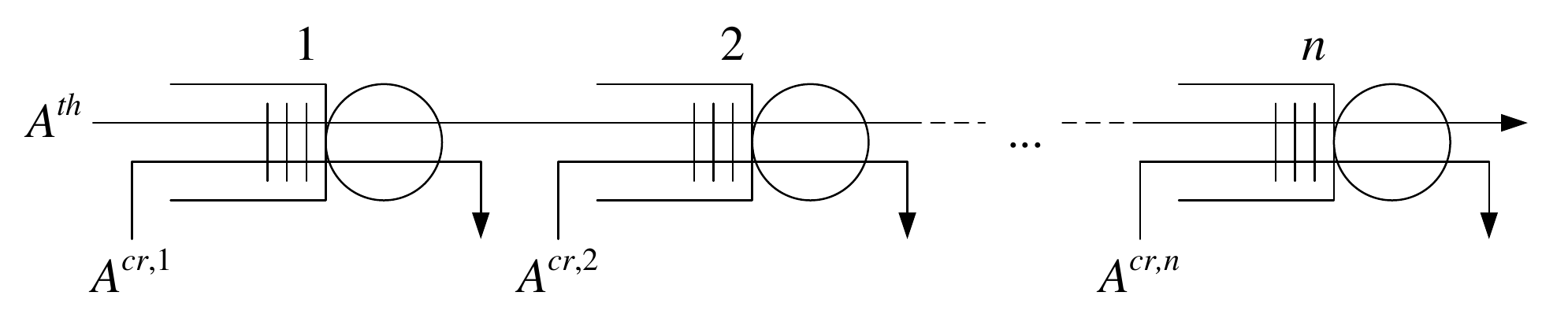}
\caption{A through flow traverses a tandem of $n$ queuing systems. Cross traffic is multiplexed and de-multiplexed at each system}
\label{fig:scenario}
\end{figure}
Cor. \ref{cor:fbme2e} enables the derivation of end-to-end service curves for networks under fBm cross traffic. In the remainder of this section we consider the line topology shown in Fig. \ref{fig:scenario} where fBm cross traffic is multiplexed and de-multiplexed at each hop. For ease of notation we assume $n$ homogeneous systems in series, each with capacity $C$ and fBm cross traffic with identical parameters $\lambda$, $\sigma$, and $H^{cr}$. We show that end-to-end performance bounds for $n$ tandem systems under fBm cross traffic grow in $\mathcal{O}\bigl(n (\log n)^{\frac{1}{2-2H}}\bigr)$.

The leftover service curve at each system is $S^i(t) = (C-r^{cr})t$ where $r^{cr}$ is the envelope rate of the fBm cross traffic. From (\ref{eq:s_e2e}) we obtain $S^{net} (t) = (C - r^{cr} - \Delta) t$ where $\Delta = (n-1) \delta$ is in $(0,C-r^{cr})$ and used as a constant. The deficit profile follows from (\ref{eq:epsilon_e2e}) by insertion of Cor. \ref{cor:fbmleftoverservice} and Cor. \ref{cor:fbme2e} as
\begin{equation*}
\varepsilon^{net}(b) \! = \! \frac{\Gamma(\frac{1}{2\beta})}{2 \vartheta^{\frac{1}{2\beta}}} \inf_{x} \Biggl\{ \!\! \frac{(n\!-\!1)^2 (\frac{b-x}{n-1})^{-\frac{1-(H+2\beta)}{\beta}}}{\Delta(1-(H+2\beta))}  + \frac{x^{-\frac{1-(H+\beta)}{\beta}}}{\beta} \! \Biggr\}
\end{equation*}
where $x \in (0,b)$ and $\beta \in \bigl(0,\frac{1-H}{2}\bigr)$.

Fig. \ref{fig:delaybnd_diff_end2end} shows end-to-end delay bounds for a CBR through flow that traverses $n$ tandem systems each with capacity $C$ and fBm cross traffic. The delay bound is computed from (\ref{eq:horizontaldeviation}) using the network service curve and CBR through traffic as $\mathsf{P}[W > b / (C - r^{cr} - \Delta)] \le \varepsilon^{net}(b)$ under the constraint that $r^{th} + r^{cr} + \Delta \le C$. The parameters of through and cross traffic are given in Tab. \ref{tab:parameters}. As before, we optimize the free parameters of the envelopes numerically.

Fig. \ref{fig:delaybnd_diff_end2end} once more confirms the huge impact of cross traffic with LRD on network performance. After numerical optimization of $\beta$ Fig. \ref{fig:delaybnd_diff_end2end} shows a slightly faster than linear growth of end-to-end delay bounds with the number of tandem systems $n$ subject to a violation probability $\varepsilon = 10^{-12}$.
\begin{figure}
\centering
\includegraphics[width=1.0\columnwidth]{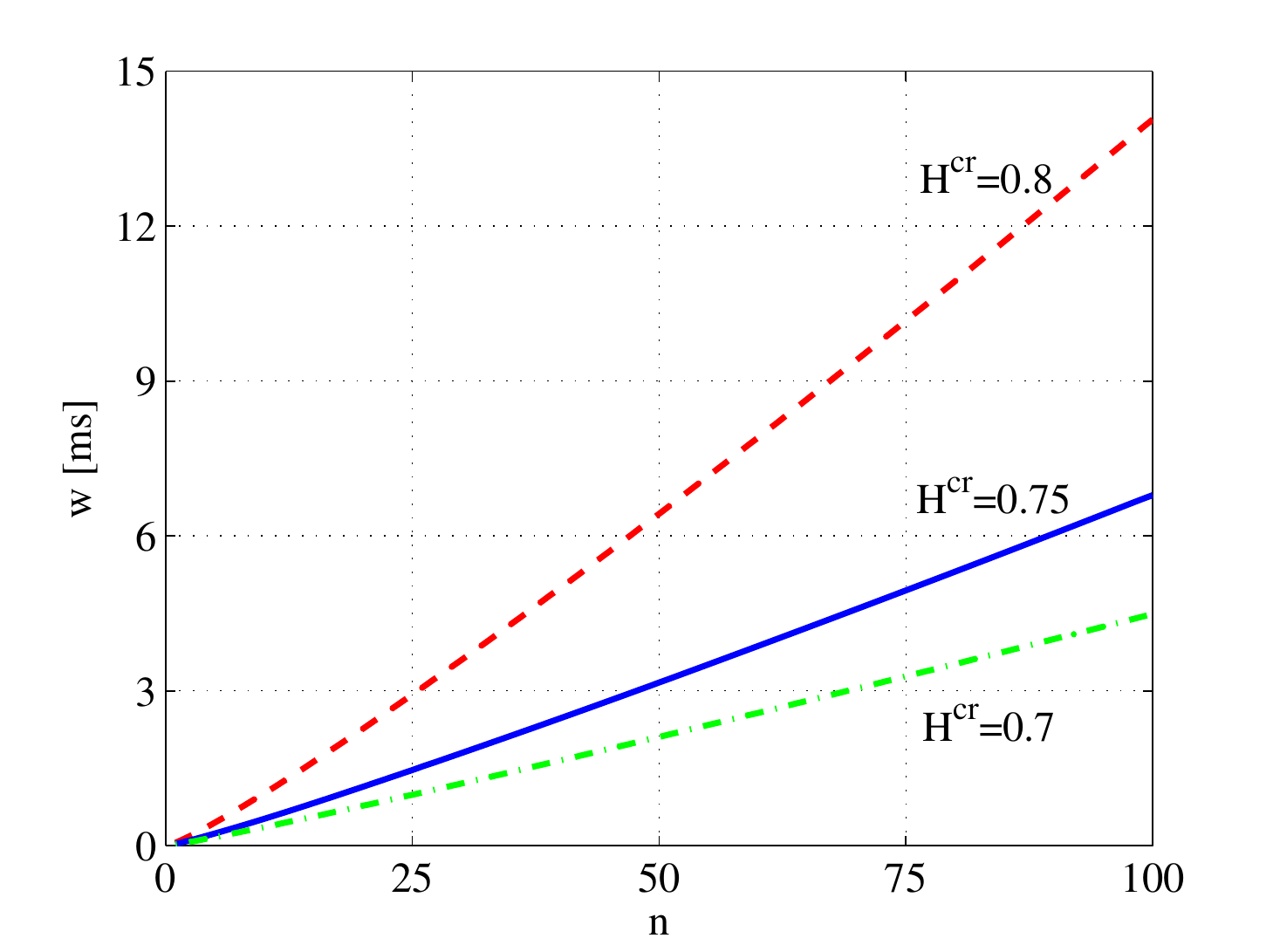}
\caption{End-to-end delay bounds with violation probability $\varepsilon = 10^{-12}$ for $n$ constant rate servers with fBm cross traffic in series. The delay bounds shown for different $H^{cr}$ of cross traffic grow slightly faster than linearly.}
\label{fig:delaybnd_diff_end2end}
\end{figure}

We derive an analytical result on the scaling of end-to-end delay bounds with the number of systems in series using Stirling's formula as in Sect. \ref{sec:nearoptimalbeta}. The growth of end-to-end delay bounds with $n$ is determined by the tail decay of the fBm overflow, respectively, deficit profiles. We neglect the irregularity in (\ref{eq:epsilon_e2e}) that is due to the last hop and estimate $\varepsilon^n$ by $\varepsilon^n_{\delta}$ to obtain the simplified deficit profile
\begin{equation*}
\varepsilon^{net}(b) = \frac{\Gamma(\frac{1}{2\beta}) n (n-1) (\frac{b}{n})^{-\frac{1-(H+2\beta)}{\beta}}}{2 \vartheta^{\frac{1}{2\beta}} \Delta(1-(H+2\beta))}.
\end{equation*}
Assuming $\beta \ll 1$ we use Stirling's formula to obtain
\begin{equation}
\tilde{\varepsilon}^{net}(b) = \frac{\sqrt{\pi \beta} n (n-1) (\frac{b}{n})^{-\frac{1-(H+2\beta)}{\beta}}}{ (2 e \beta \vartheta)^{\frac{1}{2\beta}} \Delta(1-(H+2\beta))}.
\label{eq:stirling_e2e}
\end{equation}
We use the same steps as in Sect. \ref{sec:nearoptimalbeta} to minimize $\varepsilon^{net}(b)$ over $\beta$. For $\beta \ll (1-H)/2$ we find the near optimal solution
\begin{equation*}
\beta^* = \frac{n^{2-2H}}{2(-\log\varepsilon_a)}
\end{equation*}
where $\varepsilon_a$ is defined in Cor. \ref{cor:affineenvelopes}. The term $\varepsilon_a$ is used here only to shorten notation and to express $\varepsilon^{net}$ as a multiple of $\varepsilon_a$ in the sequel. We emphasize that we do not use the approximation by the largest term (\ref{eqn:change_P_sup}).

We define $\psi = (1-H)/(1-(H+2\beta))$ and use $\chi$ from (\ref{eq:chi}) to find by insertion of $\beta^*$ into (\ref{eq:stirling_e2e}) that
\begin{equation*}
\tilde{\varepsilon}^{net}(b) = \frac{n-1}{n^H} \! \left(\frac{H}{1-H}\right)^H \! \frac{\sqrt{\pi} \psi \sigma b^{1+H}}{\Delta(r^{cr}-\lambda)^{1+H}} \varepsilon_a^{(1+\log\chi)n^{2H-2}} .
\end{equation*}
Note that both $\chi$ and $\psi$ approach 1 as $\beta \rightarrow 0$. Assuming $\beta^* \ll (1-H)/2$ we can generally find constants that bound $\chi$ and $\psi$ from below, respectively, from above. Inserting $\varepsilon_a$ as defined in Cor. \ref{cor:affineenvelopes} and using positive constants $c_i$ we obtain
\begin{equation*}
\tilde{\varepsilon}^{net}(b) \le n^{1-H} \, c_1 b^{1+H} e^{-c_2 b^{2-2H} n^{2H-2}} .
\end{equation*}
We let $b = n (c_0 \log n)^{\frac{1}{2-2H}}$ for $n \ge 2$ and find
\begin{equation*}
\tilde{\varepsilon}^{net} \le (c_0 \log n)^{\frac{1+H}{2-2H}} c_1 n^{2-c_0 c_2}.
\end{equation*}
Generally, there exists $c_0 > 2/c_2$ such that $\tilde{\varepsilon}^{net}$ is upper bounded by a constant for all $n$. Fixing $\tilde{\varepsilon}^{net}$ it follows that
\begin{equation*}
b \in \mathcal{O}\left(n (\log n)^{\frac{1}{2-2H}}\right).
\end{equation*}
Finally, we verify that $\beta^*$ decreases with $n$ where the decay is proportional to $1/\log(n)$. This confirms the assumption that given $\beta^*$ is small it also remains small with increasing $n$. Recall that the above approximation using Stirling's formula is exact for $\beta \rightarrow 0$.

For comparison we use the same network, however, with EBB instead of fBm cross traffic. As before, the EBB traffic is an aggregate of $m$ on-off sources. From the EBB sample path envelope in Sect. \ref{sec:LeftoverServiceCurveUnderFBMCrossTraffic} it follows that $S^i(t) = (C - r^{cr})t$ is a leftover service curve with deficit profile $\varepsilon(b) = e^{-\theta b}/(\theta(r^{cr}-m\rho(\theta)))$. After relaxation of the service curve by rate $\delta > 0$ the sample path deficit profile follows from (\ref{eq:departuresamplepathprofile}) as $\varepsilon_{\delta}(b) = e^{-\theta b}/ (\delta \theta^2 (r^{cr} - m \rho(\theta)))$ for any $\theta$ that satisfies $m \rho(\theta) < r^{cr}$. The network service curve under EBB cross traffic becomes $S^{net} (t) = (C - r^{cr} - \Delta) t$ with deficit profile
\begin{equation*}
\varepsilon^{net}(b) = \frac{1}{r^{cr} \!-\! m \rho(\theta)} \inf_{x} \Biggl\{ \frac{(n-1)^2 e^{-\theta \frac{b-x}{n-1}}}{\Delta \theta^2} +  \frac{e^{-\theta x}}{\theta} \Biggr\}
\end{equation*}
where $x \in (0,b)$ and $\theta$ such that $m \rho(\theta) < r^{cr}$.

Solving $\varepsilon^{net}(b)$ for $b$ yields that end-to-end performance bounds under EBB cross traffic are in $\mathcal{O}(n \log n)$. This result is derived in \cite{ciucu06} and also proven as a lower bound in \cite{burchard:scaling}. In contrast, our result for fBm cross traffic is a scaling in $\mathcal{O}\bigl(n ( \log n)^{\frac{1}{2-2H}}\bigr)$. The scaling is largely determined by the decay rate of overflow probabilities. The difference to EBB is caused by the Weibullian decay for fBm traffic as opposed to the exponential decay for EBB traffic, see Sect. \ref{sec:nearoptimalbeta}. We find that the poly-logarithmic scaling component that increases with $H$ is due to LRD. For the special case $H = \frac{1}{2}$ the overflow probability of fBm traffic decays exponentially fast and we recover the scaling $\mathcal{O} (n \log n)$ for EBB cross traffic. While we conclude that LRD has significant impact on single system performance bounds, we find that the additional effect due to concatenation of tandem systems is moderate.

\begin{figure}
\centering
\includegraphics[width=1.0\columnwidth]{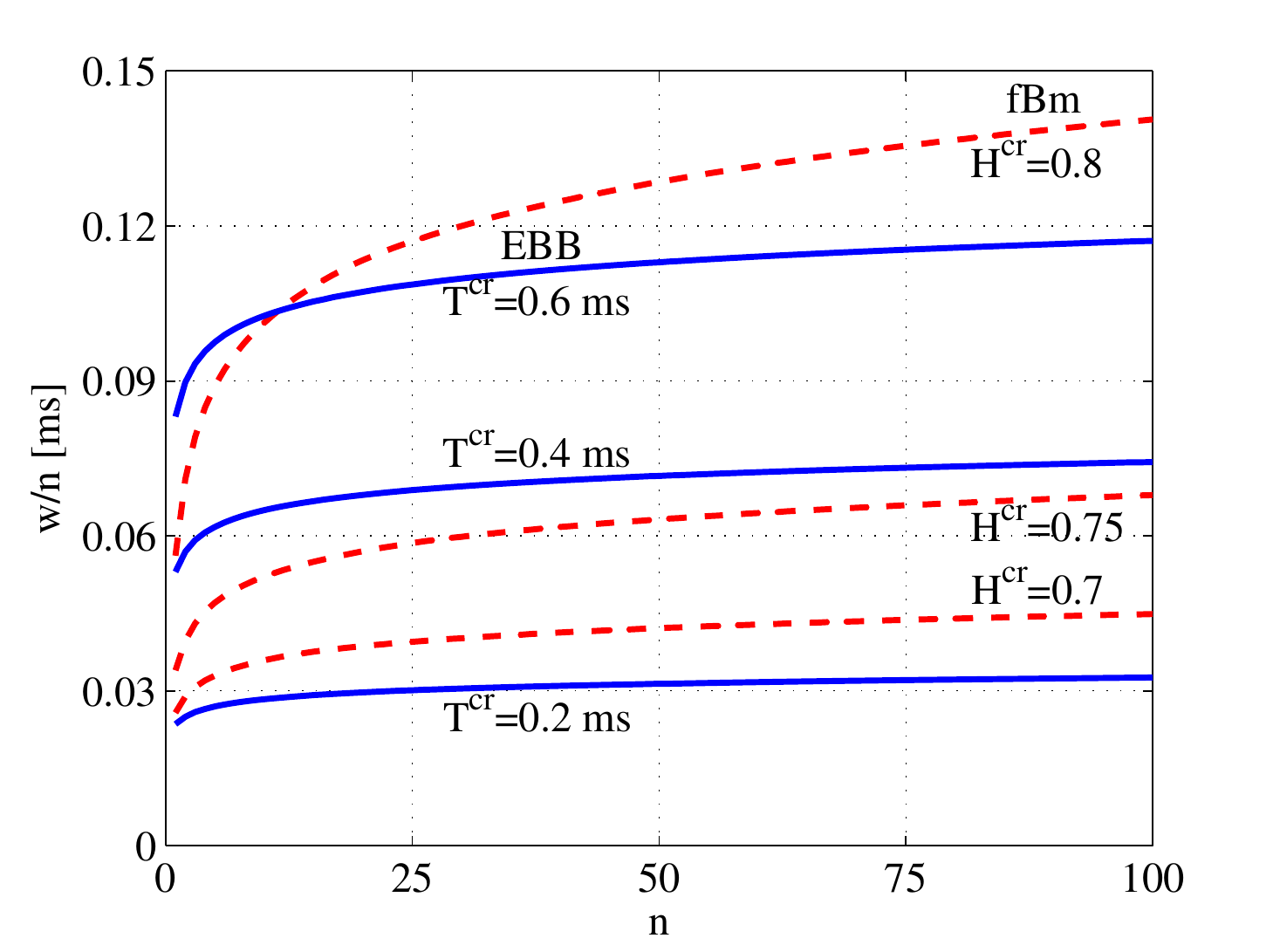}
\caption{End-to-end delay bounds as in Fig. \ref{fig:delaybnd_diff_end2end}, however, normalized by $n$ for EBB and fBm cross traffic, respectively. Under EBB cross traffic $w/n$ grows logarithmically with $n$ whereas it grows poly-logarithmically under fBm cross traffic where the exponent grows with $H$ as $1/(2-2H)$.}
\label{fig:delayinc_diff_end2end}
\end{figure}
In Fig. \ref{fig:delayinc_diff_end2end} we show end-to-end delay bounds for CBR through traffic under either EBB or fBm cross traffic. We use the same parameters and the same network as for Fig. \ref{fig:delaybnd_diff_end2end}, however, we normalize the delay by the number of systems in series $n$ to compare the logarithmic and poly-logarithmic scaling under EBB and fBm cross traffic, respectively. The traffic parameters are given in Tab. \ref{tab:parameters}. As before, we optimize the parameters of the traffic envelopes numerically. Fig. \ref{fig:delayinc_diff_end2end} clearly confirms the logarithmic growth of the normalized delay bounds. Moreover, Fig. \ref{fig:delayinc_diff_end2end} shows that a variation of the parameter $T^{cr}$, that determines the correlation of the EBB cross traffic, mainly reproduces staggered versions of $w/n$. In contrast, varying parameter $H^{cr}$ of fBm cross traffic changes the exponent of the poly-logarithmic scaling and hence alters the shape of $w/n$.

While the scaling result for EBB is in $\Theta (n \log n)$, i.e. it is proven both as an upper bound \cite{ciucu06} and a lower bound \cite{burchard:scaling}, a lower bound that complements the scaling for fBm in $\mathcal{O}\bigl(n (\log n)^{\frac{1}{2-2H}}\bigr)$ still remains to be derived. Note that the scaling results under fBm, respectively, EBB cross traffic are derived without making assumptions about statistical independence of the service left over at individual systems. Under the additional assumption of statistical independence end-to-end performance bounds that scale in $\mathcal{O}(n)$ are derived in \cite{fidler06} for $(\sigma(\theta),\rho(\theta))$ constrained cross traffic \cite{chang00} that is closely related to the EBB traffic model. The effect of statistical independence on the scaling under fBm cross traffic is not addressed in this work and is an unresolved research question. For an elaboration on known scaling results see \cite{ciucu:networkcalculusscaling}.
\section{Conclusions}
\label{sec:Conclusion}
The contribution of this paper are end-to-end statistical performance bounds for a through flow in a network under fBm cross traffic with LRD. To this end, we used the framework of the stochastic network calculus. We developed a sample path envelope for fBm traffic that complements a known approximation using the largest term. Our sample path envelope and the approximation agree in the Weibullian decay of overflow probabilities. We recovered the previous result at the point in time where the violation of an affine upper envelope by fBm traffic is most probable.

From the derived envelopes we obtained the service curve left over by fBm cross traffic at a system. By convolution of these leftover service curves we derived a network service curve. An essential intermediary result for the successful derivation is the twofold integrability of the overflow probability. We showed a numerical evaluation of the impact of fBm cross traffic on the end-to-end performance of through flows. We proved that end-to-end performance bounds for $n$ systems in series grow in $\mathcal{O}\bigl( n (\log n)^{\frac{1}{2-2H}} \bigr)$.
\section*{Acknowledgements}
This work was supported by an Emmy Noether grant of the German Research Foundation (DFG).
%

%
\section*{Appendix}
\label{sec:Appendix}
\begin{lem}[\textbf{Gamma Function}]
\label{lem:gamma}
For $x \in (0,1)$ and $\xi > 0$ it holds that
\begin{equation*}
\int^{\infty}_{0} x^{t^{\xi}} dt = \frac{\Gamma\bigl(\frac{1}{\xi}\bigr)}{\xi(-\log x)^{\frac{1}{\xi}}}
\end{equation*}
\end{lem}
\begin{proof}
The definition of the Gamma function by Euler states that for all $y > 0$
\begin{equation*}
\Gamma(y) = \int_{0}^{1} (-\log t )^{y-1} dt .
\end{equation*}
We substitute $t = e^{-z \tau^{\xi}}$ where $z > 0$ and $\xi > 0$. It follows that $dt = -z \xi e^{-z \tau^{\xi}}  \tau^{\xi-1} d\tau$ and
\begin{equation*}
\Gamma(y) = z^y \xi \int_0^{\infty} \tau^{\xi y-1} e^{-z \tau^{\xi}} d\tau .
\end{equation*}
Assuming $\xi > 0$ and letting $y = 1/\xi$ yields
\begin{equation*}
\Gamma\biggl(\frac{1}{\xi}\biggr) = z^{\frac{1}{\xi}} \xi \int_0^{\infty} e^{-z \tau^{\xi}} d\tau .
\end{equation*}
Finally, we substitute $z = - \log x$ where $x \in (0,1)$.
\end{proof}
\end{document}